\documentclass[runningheads]{llncs}

\usepackage{graphicx}
\usepackage{amssymb}
\usepackage{amsmath}
\usepackage{comment}
\usepackage[disable]{todonotes}
\usepackage[utf8]{inputenc}
\usepackage{shuffle}
\usepackage{multirow}
\usepackage{listings}
\usepackage{mathabx}
\usepackage{wrapfig}

\usepackage{thm-restate}

\usepackage{hyperref}

\usepackage{cite}

\usetikzlibrary{arrows,shapes,automata,positioning}

\newcommand{\perm}{\operatorname{perm}}

\newcommand{\lastsym}{\operatorname{last}}

\declaretheorem{Assumption}

\DeclareMathOperator{\lcm}{lcm}

 \let\llncssubparagraph\subparagraph
 \let\subparagraph\paragraph
 \usepackage[compact]{titlesec}
 \let\subparagraph\llncssubparagraph

\def\dotminus{\mathbin{\ooalign{\hss\raise1ex\hbox{.}\hss\cr
  \mathsurround=0pt$-$}}} 
 
\begin{document}\fontsize{10}{10}\rm
\title{State Complexity Bounds for the Commutative Closure of Group Languages}
%
%
\author{Stefan Hoffmann}
\authorrunning{S. Hoffmann}
%
\institute{Informatikwissenschaften, FB IV, 
  Universit\"at Trier,  Universitätsring 15, 54296~Trier, Germany, 
  \email{hoffmanns@informatik.uni-trier.de}}
\maketitle              
\begin{abstract}
 
  In this work we construct an automaton for the commutative
  closure of a  given regular group language. 
  The number of states of the resulting automaton is bounded by 
  the number of states of the original automaton, raised to the
  power of the alphabet size, times the product of the order of the letters, viewed as permutations of the state set.
  This gives the asymptotic state bound $O((n\exp(\sqrt{n\ln n}))^{|\Sigma|})$,
  if the original regular language is accepted by an automaton with $n$ states.
  Depending on the automaton in question, we label
  points of $\mathbb N_0^{|\Sigma|}$ by subsets of states and introduce unary automata
  which decompose the thus labelled grid. Based on these constructions,
  we give a general regularity condition, which is fulfilled for group languages.
  
\end{abstract}

\vspace*{-0.5cm}

\section{Introduction}

The area of state complexity asks for sharp bounds on the size of resulting automata for 
regularity-preserving operations. 
This question goes back at least to work by Maslov~\cite{Mas70}, but, starting with the work~\cite{YuZhuangSalomaa1994}, 
has revived at the end of the last millennium.
The class of deterministic and complete automata is the most natural, 
or prototypical, class. But state complexity questions have also been explored for non-deterministic automata,
or other automata models, see for example the surveys~\cite{GaoMRY17,DBLP:books/ws/p/0001K10,DBLP:journals/iandc/HolzerK11}.
As the number of states of an accepting automaton 
could be interpreted as the memory required to describe the accepted language
and is directly related to the runtime of algorithms employing regular languages, obtaining
state complexity bounds is a natural question with applications in verification, natural language
processing or software engineering \cite{DBLP:reference/crc/2012fsbma,DBLP:journals/nle/Mohri96,GaoMRY17}.
So, nowadays, it is an active and important area of reasearch under the broader theme of
descriptional complexity of systems.
We refer again to the survey \cite{GaoMRY17} for an introduction and more information. 
It was shown in \cite{DBLP:journals/iandc/GomezGP13} that the commutative closure
is regularity preserving on regular group languages.
But the method of proof 
was algebraic and used Ramsey-type arguments. The 
general form of an accepting automaton was still open.
Here we give methods to obtain such an automaton, and derive state bounds for the commutative closure of 
regular group languages.
The state complexity of the commutative closure
on finite languages was investigated
in~\cite{DBLP:conf/dcfs/PalioudakisCGHK15,DBLP:journals/tcs/ChoGHKPS17}.


\section{Prerequisites}
\label{sec::prerequisites}

Let $\Sigma = \{a_1,\ldots, a_k\}$ be a finite set of symbols\footnote{If not otherwise stated we assume that our alphabet has
the form 
 $\Sigma = \{a_1, \ldots, a_k\}$, and $k$ denotes the number
 of symbols.},
 called an alphabet. The set $\Sigma^{\ast}$ denotes
the set of all finite sequences, i.e., of all words. The finite sequence of length zero,
or the empty word, is denoted by $\varepsilon$. For a given word $w$ we denote by $|w|$
its length, and, for $a \in \Sigma$, by $|w|_a$ the number of occurrences of the symbol $a$
in $w$. Subsets of $\Sigma^{\ast}$
are called languages. With $\mathbb N_0 = \{ 0,1,2,\ldots \}$
we denote the set of natural numbers, including zero. 
A finite deterministic and complete automaton will
be denoted by $\mathcal A = (\Sigma, S, \delta, s_0, F)$,
with $\delta : S \times \Sigma \to S$ the state transition function, $S$ a finite set of states, $s_0 \in S$
the start state and $F \subseteq S$ the set of final states. 
The properties of being deterministic and complete are implied by the definition of $\delta$
as a total function.
The transition function $\delta : S \times \Sigma \to S$
could be extended to a transition function on words $\delta^{\ast} : S \times \Sigma^{\ast} \to S$,
by setting $\delta^{\ast}(s, \varepsilon) := s$ and $\delta^{\ast}(s, wa) := \delta(\delta^{\ast}(s, w), a)$
for $s \in S$, $a \in \Sigma$ and $w \in \Sigma^{\ast}$. In the remainder we drop
the distinction between both functions and will also denote this extension by $\delta$.
The language accepted by some automaton $\mathcal A = (\Sigma, S, \delta, s_0, F)$ is
$
 L(\mathcal A) = \{ w \in \Sigma^{\ast} \mid \delta(s_0, w) \in F \}.
$
A language $L \subseteq \Sigma^{\ast}$ is called regular if $L = L(\mathcal A)$
for some finite automaton. 
The \emph{state complexity} of a regular language is the size of a minimal automaton accepting
this language.
An automaton is called a \emph{permutation automaton}
if the transformation of the states induced by a letter
is a permutation, i.e., a bijective function. A regular language
is called a \emph{group language} if it is accepted by some permutation automaton.
The map $\psi : \Sigma^{\ast} \to \mathbb N_0^k$
given by $\psi(w) = (|w|_{a_1}, \ldots, |w|_{a_k})$ is called the \emph{Parikh morphism}.\todo{Ab hier schmeißt Du phi uns psi durcheinander. Da psi Standard für Parikh ist, solltest Du für dein Zustandslabeln vlt. phi nehmen, oder vlt auch sigma

SH: Vereinheitlich. $\psi$ für Parikh map, $\sigma_{\mathcal A}$ für state-label function.}
For a given word $w \in \Sigma^{\ast}$ we define the \emph{commutative closure}
as $\perm(w) := \{ u \in \Sigma^{\ast} : \psi(u) = \psi(w) \}$.
For languages $L \subseteq \Sigma^{\ast}$ we set $\perm(L) := \bigcup_{w\in L} \perm(w)$.
A language is called \emph{commutative}
if $\perm(L) = L$, i.e., with every word each permutation of this word is also in the language.
Every function $f : X \to Y$ could be extended to subsets $S \subseteq X$ by setting $f(S) := \{ f(x) : x \in S \}$, we will do this frequently without special mentioning. For $Z \subseteq X$ we denote by $f_{|Z} : Z \to Y$ 
the function obtained 
by restriction of the arguments to elements of $Z$.
For a set $X$, we denote by $\mathcal P(X) = \{ Y : Y \subseteq X \}$
the power set of $X$. If $X,Y$ are sets, by $X \times Y$ we denote their cartesian product.
By $\pi_1 : X\times Y \to X$ and $\pi_2 : X \times Y \to Y$ we denote
the projection maps onto the first and second component, $\pi_1(x,y) = x$ and $\pi_2(x,y) = y$.
If $a,b\in \mathbb N_0$ with $b > 0$, we denote by $a \bmod b$ the unique number $0 \le r < b$
such that $a = bn + r$ for some $n \ge 0$.
For $n \in \mathbb N_0$ we set $[n] := \{ k \in \mathbb N_0 : 0 \le k < n \}$.
Let $M \subseteq \mathbb N_0$ be some \emph{finite} set. 
By $\max M$ we denote the maximal element in $M$
with respect to the usual order, and we set $\max \emptyset = 0$.
Also for finite $M \subseteq \mathbb N_0 \setminus\{0\}$, i.e., $M$ is finite without zero in it,
by $\lcm M$ we denote the
least common multiple of the numbers in $M$, and set $\lcm \emptyset = 0$.

\subsection{Unary Languages}
\label{sec:unary}

Let $\Sigma = \{a\}$ be a unary alphabet. In this section we collect some results about unary languages.
Suppose $L \subseteq \Sigma^{\ast}$ is regular
with an accepting complete deterministic automaton $\mathcal A = (\Sigma, S, \delta, s_0, F)$. Then, by considering
the sequence of states $\delta(s_0, a^1), \delta(s_0, a^2), \delta(s_0, a^3), \ldots$ we find numbers $i \ge 0, p > 0$ with $i$ and $p$ minimal such that $\delta(s_0, a^i) = \delta(s_0, a^{i+p})$.
We call these numbers the index $i$ and the period $p$ of the automaton $\mathcal A$.
Suppose $\mathcal A$ is initially connected, i.e., $\delta(s_0, \Sigma^*) = Q$.
Then $i + p = |S|$, the states from $\{ s_0, \delta(s_0, a), \ldots, \delta(s_0, a^{i-1}) \}$
constitute the \emph{tail}, and the states
from $\{ \delta(s_0, a^i), \delta(s_0, a^{i+1}), \ldots, \delta(s_0, a^{i+p-1} \}$
constitute the unique \emph{cycle} of the automaton.
If $\mathcal A$ is not initially connected, when we speak of the cycle or tail
of that automaton, we nevertheless mean the above sets, despite the automaton graph
might have more than one cycle, or more than one straight path.

\section{Results} 

\subsection{Intuition, Method of Proof and Main Results}



We have two main results, first a general automaton construction for the commutative closure of a regular
group language, and second a more general framework to derive this result, which entails a general
regularity condition for commutative closures. The first result, in asymptotic form.

  \begin{theorem}{(Asymptotic version)}
  \label{thm:state_compl_grp_lang}
      For a regular group language with state complexity $n$,
      the comm. closure is regular with
      state complexity in $O((n e^{\sqrt{n \ln n}})^{|\Sigma|})$. 
      \end{theorem}
      
In Theorem~\ref{thm:sc_group_case}, a more quantitative version in terms of the constructions will be given.
But this result is more an application of a general scheme, which 
will be useful in future investigations
as well. So, let us spend some time in explaining the basic idea.
The constructions and definitions
that follow in the next sections 
are rather involved and technical, but they are, I hope to convince the reader, the worked out formalisation of quite a natural idea. Imagine you have an operation that identifies certain words (in our case, we identify words if they
are permutations of each other) and you apply this operation to a regular language. How could an (hitherto possibly infinite state) automaton for the result of this operation look? If you have two words $u$ and $v$, which are identified and drive the original automaton into two states, say $s$ and $t$, then in what state should an automaton end up for the resulting language after this identification? As it should not distinguish between both words, as they are identified, a possible state is $\{u,v\} \times \{s,t\}$, 
the set $\{u,v\}$ represents the read in word under identification, and $\{s,t\}$ represents
the possible states of the original automaton.
Applied to our situation, words are identified if they have the same Parikh image, i.e., the letter counts are equals. 
So, we start with \emph{labelling the grid $\mathbb N_0^k$ with the states that are reachable
by words whose Parikh image equals the point in question}.
Hence, if the original automaton has state set $Q$, we can think of as constructing an (infinite) automaton 
with state set $\mathbb N_0^k \times \mathcal P(Q)$. We will not formally construct this automaton, but it is implicit in the constructions we give. This automaton could be used to accept the commutative closure, where a word is accepted
if, after reading this word,  the second component, the \emph{state set label}, contains at least one final state, meaning for some word, which is equivalent to the read in word, we can reach a final state.
For a regular language, this construction could also be viewed as a generalized Parikh map, where
we not only label a point by the binary information if some word with that letter count is in the language or not,
but we have the more rich information what states are reachable by all permutations of a given word.
By only looking at the state labels containing a final state, we can recover the original Parikh map.
We will adopt this viewpoint, which is sufficient for the results, in the formal treatment to follow.
Also, note we have an intuitive correspondence to the power set construction, in the sense that
in this construction, the states, as sets, save all possibilities to end up after reading a word. 
Here, our state labels serve the same purpose as saving all possiblities.
So, intuitively and very roughly, the method
   could be thought of as both a refined Parikh map for regular languages
   and a power set construction for automata that incorporates the commutativity
   condition.
   
It turns out that the story does not end here, but that the ``generalized Parikh map'', or \emph{state (label) map},
as we will call it in the following, admits a lot of structure that allows to derive a regularity
criterion. 
First, let us state our regularity condition in intuitive terms, a more refined
statement is given as Theorem~\ref{thm:regularity_condition} later.

\begin{theorem}{(Intuitive form)}
\label{thm:regularity_condition_intuitive}
 Suppose the grid $\mathbb N_0^k$ is labelled by the states of a given automaton.
 If we have a universal bound $N \ge 0$ and a period $P > 0$
 such that, for $p = (p_1, \ldots, p_k) \in \mathbb N_0^K \setminus ([N] \times \ldots \times [N])$
 and $j \in \{1,\ldots, k\}$, the labels at $p$
 and $(p_1, \ldots, p_{j-1}, p_j - P, p_{j+1}, \ldots, p_k)$
 are equal, then the commutative closure of the 
 language described by the original automaton is regular and could be accepted
 by an automaton of size at most $N^k$.
\end{theorem}


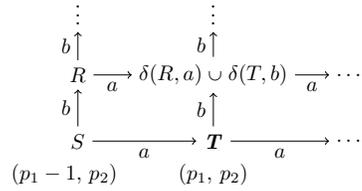
\begin{wrapfigure}[27]{r}{0.4\textwidth}
 \centering
  \scalebox{.8}{
    \begin{tikzpicture}[yscale=1.1, xscale=2.25] 

        \node at (1, -0.5)    {($p_1$, $p_2$)};
        \node at (-0.1, -0.5) {($p_1-1$, $p_2$)};
      
        \node at (0,0) (0p0) {$S$};
        \node at (1,0) (1p0) {\textbf{\emph{T}}};
        \node at (2,0) (2p0) {$\ldots$};
         
        \node at (0,1) (0p1) {$R$};
        \node at (0,2) (0p2) {$\vdots$};

        \node at (1,1) (1p1) {$\delta(R,a) \cup \delta(T, b)$};
        \node at (2,1) (2p1) {$\ldots$};  
        \node at (1,2) (1p2) {$\vdots$};  
        
       \path[->] (0p0) edge [below] node {$a$} (1p0)
                 (1p0) edge [below] node {$a$} (2p0);
                 
       \path[->] (0p0)  edge [left] node {$b$} (0p1)
                 (0p1)  edge [left] node {$b$} (0p2);
                 
       \path[->] (0p1) edge [below] node {$a$} (1p1);
       \path[->] (1p0) edge [left]  node {$b$} (1p1);
       \path[->] (1p1) edge [left]  node {$b$} (1p2);
       \path[->] (1p1) edge [below]  node {$a$} (2p1);
    \end{tikzpicture}}
  \caption{\footnotesize In illustration of how state labels are updated if new input symbols are read.
   We are at $(p_1, p_2)$ with state label $T$ and read the input $b$. So, we 
   will end up at $(p_1, p_2+1)$. Then
   the state label, for an automaton $\mathcal A = (\Sigma, Q, \delta, s_0, F)$,
   at $(p_1, p_2+1)$ is made up out of the state label $T$, but also out of the neighboring
   state label $R$. Imagine as ``going in both'' ways, i.e., the path $ab$ and $ba$,
   to compute the new state label. Please see the text for more explanation.}
    \label{fig:state_label_update}
\end{wrapfigure}

The basic mechanism behind this theorem is a decomposition of the state labels into unary automata.
I must confess, the form of these unary automata, formally stated in Definition~\ref{def:sequ_grid_decomp_aut}, is rather involved and, to be honest, took me quite some time to come up with.
The idea is to implement in these automata how the state labels are influenced by neighboring state labels. Imagine
we are at a certain point, then we receive an input letter and go to the next point that
corresponds to this additional input letter. What should the new state label look like?
First, we should carry with us the state label from the previous point, but updated with
the input letter. But, by the nature of the commutation relation, this input letter could
also be read at a previous point up to permutational identification of words.
So, we also go back to our previous point and investigate all the state labels of its neighboring
points from which we could reach this point. It turns out that it is enough to look at those points from
which we could reach the current point in one step. Then, we took their state labels, update them
for the input letter by going into the direction of this letter, but, after this, also
take the state label thus obtained back to our target state label in correspondence
with the letter from which we got from the neighboring state to the starting state.
We go around ``both commuting'' letters, please see Figure~\ref{fig:state_label_update}
for the case of $\Sigma = \{a,b\}$. It will be shown that this
operational scheme could be implemented into unary automata, which in this sense decompose 
the state labelling.
So, let us take the journey and see how these ideas are actually implemented!

 \subsubsection{Outline: } 
 
 In Section \ref{sec:reg_cond}, we first give a labelling of the grid $\mathbb N_0^k$
   by states of a given automaton. This labelling is
   in some sense an abstract description of the commutative closure, which is more precisely
   stated in Corollary \ref{cor:state_label_described_com_closure}.
   We then construct unary automata for each letter. Very roughly, and intuitively, they read
   in letters parallel to the direction of this letter in $\mathbb N_0^k$, given by the Parikh map.
   We have one such automaton for each point on the hyperplane orthogonal to this direction.
   These unary automata are then used to decribe the mentioned state labelling.
   In this sense, the state labelling is decomposed into these automata. This is made
   more precise in Proposition \ref{prop:grid_aut_decomp}.
   If all the automata in this decomposition, for each letter, only have a bounded number 
   of states, then the commutative closure is a regular language. 
   By using the indices and periods, we give a state bound for the resulting
   automaton in Theorem \ref{thm:regularity_condition}.
   In Section \ref{sec:grp_case}, these results are applied to the case that the given automaton is
   a permutation automaton. It turns out, stated in Proposition \ref{prop:bound_periods_Ap_group}
   and Proposition \ref{prop:index_specific_set}, that the index and the period
   are always bounded, for a bound dependent on the input automaton, which is also stated
   in these Propositions. Intuitively, the main observations why this works is that
   1) for permutations, the state labels cannot decrease as the unary automata read in symbols, and
   2) we know when the state labels must become periodic.
   Finally, applying our general result, then gives that
   the commutative closure is regular, and also yields a state complexity bound.


    \subsection{A Regularity Condition by Decomposing into Unary Automata}
    \label{sec:reg_cond}
    
     First, we introduce the state label function.
    
  \begin{definition}{(state label function)} \label{def:state-label-function}
   Let $\Sigma = \{a_1, \ldots, a_k\}$ be the alphabet. Suppose $\mathcal A = (\Sigma, Q, \delta, s_0, F)$
   is a finite automaton. The state label function, associated
   to the automaton, is the function $\sigma_{\mathcal A} : \mathbb N_0^{|\Sigma|} \to \mathcal P(Q)$
   given by 
   $$
    \sigma_{\mathcal A}(p) = \{ \delta(s_0, u) : \psi(u) = p) \}.
   $$
  \end{definition}
  
  The value of the function $\sigma_{\mathcal A}$, for some fixed automaton
  $\mathcal A = (\Sigma, Q, \delta, s_0, F)$, will also be called
  the \emph{state (set) label} for that point, or the \emph{state set} corresponding to that point.
  
  \begin{example}\label{ex:grid_aut} \footnotesize
 Consider the minimal automaton for the language\footnote{Here the minimal automaton
 has the property that no word induces a non-trivial permutation on some subset of states. Languages
 which admit such automata are called aperiodic in the literature. In some sense these are contrary 
 to group languages, the class considered in this paper.} $(a_1 a_2)^*$. The commutative
 closure of this language is not regular, as it is precisely the language of words with an equal number
 of both symbols.
 \begin{figure}[ht]
\begin{minipage}{0.4\textwidth}
\scalebox{.85}{
\begin{tikzpicture}[>=latex',shorten >=1pt,node distance=2cm,on grid,auto]
 \node[state, initial, accepting] (1) {$s_0$};
 \node[state]                     (3) [above right of=1] {$s_2$};
 \node[state]                     (2) [below right of=3] {$s_1$};
 \path[->] (1) edge [bend left=20] node {$a_1$} (2);
 \path[->] (2) edge [bend left=20] node {$a_2$} (1);
 \path[->] (1) edge [left] node {$a_2$} (3); 
 \path[->] (2) edge [right] node {$a_1$} (3);
 \path[->] (3) edge [loop above] node {$a_1, a_2$} (3);
  
 \node at (1.2,-1.5) {$Q = \{s_0, s_1, s_2\}$};
\end{tikzpicture}}
\end{minipage}%
\begin{minipage}{0.55\textwidth}
\scalebox{.8}{
    \begin{tikzpicture}[yscale=1.1, xscale=1.45] 
        \node at (-0.1, -0.5) {(0, 0)};
        
        \node at (0,0) (0p0) {$\{s_0\}$};
        \node at (1,0) (1p0) {$\{s_1\}$};
        \node at (2,0) (2p0) {$\{s_2\}$};
        \node at (3,0) (3p0) {$\{s_2\}$};
        \node at (4,0) (4p0) {$\{s_2\}$};
        \node at (5,0) (5p0) {$\ldots$};
         
        \node at (0,1) (0p1) {$\{s_2\}$};
        \node at (0,2) (0p2) {$\{s_2\}$};
        \node at (0,3) (0p3) {$\{s_2\}$};
        \node at (0,4) (0p4) {$\{s_2\}$};
        \node at (0,5) (0p5) {$\vdots$};
        \node at (1,5) (1p5) {$\vdots$};
        \node at (2,5) (2p5) {$\vdots$};
        \node at (3,5) (3p5) {$\vdots$};
        \node at (4,5) (4p5) {$\vdots$};

        \node at (1,1) (1p1) {$\{s_0,s_2\}$};
        \node at (2,1) (2p1) {$\{s_1,s_2\}$};
        \node at (3,1) (3p1) {$\{s_2\}$};
        \node at (4,1) (4p1) {$\{s_2\}$};
        \node at (5,1) (5p1) {$\ldots$};  
        
        \node at (1,2) (1p2) {$\{s_2\}$};
        \node at (2,2) (2p2) {$\{s_0,s_2\}$};
        \node at (3,2) (3p2) {$\{s_1,s_2\}$};
        \node at (4,2) (4p2) {$\{s_2\}$};
        \node at (5,2) (5p2) {$\ldots$};  
        
        \node at (1,3) (1p3) {$\{s_2\}$};
        \node at (2,3) (2p3) {$\{s_2\}$};
        \node at (3,3) (3p3) {$\{s_0,s_2\}$};
        \node at (4,3) (4p3) {$\{s_1,s_2\}$};
        \node at (5,3) (5p3) {$\ldots$}; 
        
        \node at (1,4) (1p4) {$\{s_2\}$};
        \node at (2,4) (2p4) {$\{s_2\}$};
        \node at (3,4) (3p4) {$\{s_2\}$};
        \node at (4,4) (4p4) {$\{s_0,s_2\}$};
        \node at (5,4) (5p4) {$\ldots$};

       \path[->] (0p0) edge [below] node {$a_1$} (1p0)
                 (1p0) edge [below] node {$a_1$} (2p0)
                 (2p0) edge [below] node {$a_1$} (3p0)
                 (3p0) edge [below] node {$a_1$} (4p0)
                 (4p0) edge [below] node {$a_1$} (5p0);
                 
       \path[->] (0p0)  edge [left] node {$a_2$} (0p1)
                 (0p1)  edge [left] node {$a_2$} (0p2)
                 (0p2)  edge [left] node {$a_2$} (0p3)
                 (0p3)  edge [left] node {$a_2$} (0p4)
                 (0p4)  edge [left] node {$a_2$} (0p5);
                 
       \path[->] (0p1) edge [below] node {$a_1$} (1p1)
                 (1p1) edge [below] node {$a_1$} (2p1)
                 (2p1) edge [below] node {$a_1$} (3p1)
                 (3p1) edge [below] node {$a_1$} (4p1)
                 (4p1) edge [below] node {$a_1$} (5p1);
                 
       \path[->] (0p2) edge [below] node {$a_1$} (1p2)
                 (1p2) edge [below] node {$a_1$} (2p2)
                 (2p2) edge [below] node {$a_1$} (3p2)
                 (3p2) edge [below] node {$a_1$} (4p2)
                 (4p2) edge [below] node {$a_1$} (5p2);
                 
       \path[->] (0p3) edge [below] node {$a_1$} (1p3)
                 (1p3) edge [below] node {$a_1$} (2p3)
                 (2p3) edge [below] node {$a_1$} (3p3)
                 (3p3) edge [below] node {$a_1$} (4p3)
                 (4p3) edge [below] node {$a_1$} (5p3);
                 
       \path[->] (0p4) edge [below] node {$a_1$} (1p4)
                 (1p4) edge [below] node {$a_1$} (2p4)
                 (2p4) edge [below] node {$a_1$} (3p4)
                 (3p4) edge [below] node {$a_1$} (4p4)
                 (4p4) edge [below] node {$a_1$} (5p4);
                 
       \path[->] (1p0)  edge [left] node {$a_2$} (1p1)
                 (1p1)  edge [left] node {$a_2$} (1p2)
                 (1p2)  edge [left] node {$a_2$} (1p3)
                 (1p3)  edge [left] node {$a_2$} (1p4)
                 (1p4)  edge [left] node {$a_2$} (1p5);
                 
       \path[->] (2p0)  edge [left] node {$a_2$} (2p1)
                 (2p1)  edge [left] node {$a_2$} (2p2)
                 (2p2)  edge [left] node {$a_2$} (2p3)
                 (2p3)  edge [left] node {$a_2$} (2p4)
                 (2p4)  edge [left] node {$a_2$} (2p5);
                 
       \path[->] (3p0)  edge [left] node {$a_2$} (3p1)
                 (3p1)  edge [left] node {$a_2$} (3p2)
                 (3p2)  edge [left] node {$a_2$} (3p3)
                 (3p3)  edge [left] node {$a_2$} (3p4)
                 (3p4)  edge [left] node {$a_2$} (3p5);
                 
       \path[->] (4p0)  edge [left] node {$a_2$} (4p1)
                 (4p1)  edge [left] node {$a_2$} (4p2)
                 (4p2)  edge [left] node {$a_2$} (4p3)
                 (4p3)  edge [left] node {$a_2$} (4p4)
                 (4p4)  edge [left] node {$a_2$} (4p5);
    \end{tikzpicture}}
\end{minipage}
   \caption{\footnotesize The minimal automaton of $(a_1 a_2)^*$ and a resulting state
   labelling in $\mathbb N_0^k$. Compare this to Example \ref{ex:perm_aut},
   where the labelling is given by a permutation automaton. The final state
   set is marked by a double circle. See Example \ref{ex:grid_aut} for an explanation.}
    \label{fig::example}
\end{figure}
\end{example}

  The image of the Parikh morphism could be described by
  the state label function. In this sense, for a fixed regular language, it is a more finer notion
  of the Parikh image. 
  
  \begin{restatable}[]{proposition}{connparikh}\label{prop:conn_parikh}
 {(Connection with Parikh morphism)}
 Assume $\Sigma = \{a_1, \ldots, a_k\}$. Let $\psi : \Sigma^* \to \mathbb N_0^k$ be
 the Parikh morphism. Suppose $\mathcal A = (\Sigma, Q, \delta, s_0, F)$
 is a finite automaton. Let $\sigma_{\mathcal A} : \mathbb N_0^k \to \mathcal P(Q)$
 be the state label function. Then
  $$
   \psi( L(\mathcal A) ) = \sigma_{\mathcal A}^{-1}(\{ S \subseteq Q \mid S \cap F \ne \emptyset\}). 
  $$
 \end{restatable}
 
  As $\perm(L) = \psi^{-1}(\psi(L))$ for every language $L \subseteq \Sigma^*$,
  the next is implied.
  
  \begin{corollary}\label{cor:state_label_described_com_closure}
   Let $\Sigma = \{a_1, \ldots, a_k\}$ be our alphabet. Denote by $\psi : \Sigma^* \to \mathbb N_0^k$
   the Parikh morphism. Suppose $\mathcal A = (\Sigma, Q, \delta, s_0, F)$
   is a finite automaton. Let $\sigma_{\mathcal A} : \mathbb N_0^k \to \mathcal P(Q)$
   be the state label function. Then
   $$
   \perm(L(\mathcal A)) = \psi^{-1}(\sigma_{\mathcal A}^{-1}(\{ S \subseteq Q \mid S \cap F \ne \emptyset\})).
   $$
  \end{corollary}
   
  Next, we introduce a notion for the hyperplanes that we will use in Definition~\ref{def:sequ_grid_decomp_aut}.

    \begin{definition}{(hyperplane aligned with letter)}\label{def:hyperplance}
     Let $\Sigma = \{a_1, \ldots, a_k\}$ and $j \in \{1,\ldots, k\}$.
     We set
     \begin{equation}\label{eqn:hyperplane}
      H_j = \{ (p_1, \ldots, p_k) \in \mathbb N_0^k \mid p_j = 0\}.
     \end{equation}
    \end{definition}
    
    Suppose $\Sigma = \{a_1, \ldots, a_k\}$ and $j \in \{1,\ldots, k\}$.
    We will decompose the state label map into unary automata.
    For each letter $a_j$ and point $p \in H_j$, we construct unary automata $\mathcal A_p^{(j)}$.
    They are meant to read in inputs in the direction $\psi(a_j)$, which is orthogonal
    to $H_j$. This will be stated more precisely in Proposition \ref{prop:grid_aut_decomp}.
    
    \begin{definition}{(unary automata along letter $a_j \in \Sigma$)}
    \label{def:sequ_grid_decomp_aut}
     Let $\Sigma = \{a_1, \ldots, a_k\}$. Suppose 
     $\mathcal A = (\Sigma, Q, \delta, s_0, F)$ is a finite automaton.
     Fix  $j \in \{1,\ldots, k\}$ and $p \in H_j$.
     We define a unary automaton $\mathcal A_p^{(j)} = (\{a_j\}, Q_p^{(j)}, \delta_p^{(j)}, s_p^{(0,j)}, F_p^{(j)})$. But suppose for points $q \in \mathbb N_0^k$
     with $p = q + \psi(b)$ for some $b \in \Sigma$
     the unary automata $\mathcal A_q^{(j)} = (\{a_j\}, Q_q^{(j)}, \delta_q^{(j)}, s_q^{(0,j)}, F_q^{(j)})$ are already defined.
     Set\footnote{Note that in the definition of $\mathcal P$,
     as $p \in H_j$, we have $b \ne a_j$ and $q \in H_j$. In general, points $q \in \mathbb N_0^k$
     with $p = q + \psi(b)$, for some $b \in \Sigma$, are  predecessor points in the grid $\mathbb N_0^k$.}
     $$
      \mathcal P =\{ \mathcal A_q^{(j)} \mid p = q + \psi(b) \mbox{ for some } b \in \Sigma \}.
     $$ 
     Let $I$ and $P$ be the maximal index and the least common multiple\footnote{Note $\max\emptyset = 0$
     and $\lcm\emptyset = 1$.} of the periods of the unary automata in $\mathcal P$.
     Then set
     \begin{align} 
        Q_p^{(j)} & = \mathcal P(Q) \times [I+P] \nonumber  \\
        s_p^{(0,j)} & = (\sigma_{\mathcal A}(p),0) \label{eqn:def_unary_aut_start_state} \\ 
        \label{eqn:def_unary_aut_transition}
        \delta_p^{(j)}( (S, i), a_j ) & =
       \left\{ \begin{array}{ll}
         (T, i + 1) & \mbox{ if } i+1 < I+P \\ 
         (T, I)     & \mbox{ if } i+1 = I+P.
       \end{array}\right.
    \end{align}
     where
    \begin{equation}\label{eqn:def_seq_grid_aut_transition}
      T = \delta(S, a_j)
           \cup \bigcup_{\substack{(q, b) \in \mathbb N_0^k \times\Sigma\\ p = q + \psi(b)}} \delta(\pi_1(\delta_q^{(j)}(s_q^{(0,j)}, a_j^{i+1})), b)
    \end{equation}
    and $F_p^{(j)} = \{ (S, i) \mid S \cap F \ne \emptyset \}$.
    For a state $(S, i) \in Q_p^{(j)}$, the set $S$ will be called
    the \emph{state (set) label}, or the \emph{state set associated with it}.
    \end{definition}

   See Example \ref{ex:seq_aut} for concrete constructions of the automata from Definition \ref{def:sequ_grid_decomp_aut}.
    
    \begin{example}\label{ex:seq_aut}\footnotesize
     In Figure \ref{ex:seq_aut} we list the reachable part from the start state of the unary automata $\mathcal A_{(0,0)}^{(2)}$,
     $\mathcal A_{(1,0)}^{(2)}$, $\mathcal A_{(2,0)}^{(2)}$ and $\mathcal A_{(3,0)}^{(2)}$
      corresponding to the automaton from Example \ref{ex:grid_aut} in order.
      Each automaton is constructed from  previous ones 
      according to Definition~\ref{def:sequ_grid_decomp_aut}.
      Note that, for example for $\mathcal A_{(1,0)}^{(2)}$, the state label of the second state is the union
     of the action of $a_2$ on $\{s_1\}$, i.e, the set $\delta(\{s_1\}, a_2)$,
     but also of $a_1$ on the state label $\{s_2\}$
     of the second 
     state of the previous automaton $\mathcal A_{(0,0)}^{(2)}$. 
     Note also that the second ''counter'' component is not enough to determine all states, 
     as at the end some automata have equal values in this entry (this is essentially how these automata
     grow in size).
 \begin{figure}[htb]
     \scalebox{.8}{ 
\begin{tikzpicture}[>=latex',shorten >=1pt,node distance=2.5cm,on grid,auto]
 \tikzset{elliptic state/.style={draw,rectangle}}
 \node[elliptic state] (1) {$\{ s_0 \}$};
 \node[elliptic state]                     (3) [right of=1] {$\{ s_2 \}$};
 \path[->] (1) edge node {$a_2$} (3);
 \path[->] (3) edge [loop right] node {$a_2$} (3);
 \node at (-1.3,0) {$\mathcal A_{(0,0)}^{(2)}$};
\end{tikzpicture}}

   \scalebox{.85}{
\begin{tikzpicture}[>=latex',shorten >=1pt,node distance=2.5cm,on grid,auto]
 \tikzstyle{state}=[draw,rectangle]
 \node[state] (1) {$(\{ s_1 \},0)$};
 \node[state]                     (2) [right of=1] {$(\{ s_0, s_2 \},1)$};
 \node[state]                     (3) [right of=2] {$(\{ s_2 \},1)$};
 \path[->] (1) edge node {$a_2$} (2);
 \path[->] (2) edge node {$a_2$} (3);
 \path[->] (3) edge [loop right] node {$a_2$} (3);
 \node at (-1.3,0) {$\mathcal A_{(1,0)}^{(2)}$};
\end{tikzpicture}}

   \scalebox{.85}{
\begin{tikzpicture}[>=latex',shorten >=1pt,node distance=2.5cm,on grid,auto]
 \tikzstyle{state}=[draw,rectangle]
 \node[state] (1) {$(\{ s_2 \},0)$};
 \node[state]                     (2) [right=of 1] {$(\{ s_1, s_2 \},1)$};
 \node[state]                     (4) [right=of 2] {$(\{ s_0, s_2 \},2)$};
 \node[state]                     (3) [right=of 4] {$(\{ s_2 \},2)$};
 \path[->] (1) edge node {$a_2$} (2);
 \path[->] (2) edge node {$a_2$} (4);
 \path[->] (4) edge node {$a_2$} (3);
 \path[->] (3) edge [loop right] node {$a_2$} (3);
 \node at (-1.3,0) {$\mathcal A_{(2,0)}^{(2)}$};
\end{tikzpicture}}

   \scalebox{.85}{
\begin{tikzpicture}[>=latex',shorten >=1pt,,node distance=2.5cm,on grid,auto]
 \tikzstyle{state}=[draw,rectangle]
 \node[state] (1) {$(\{ s_2 \},0)$};
 \node[state]                     (2) [right=of 1] {$(\{ s_2 \},1)$};
 \node[state]                     (3) [right=of 2] {$(\{ s_1, s_2 \},2)$};
 \node[state]                     (4) [right=of 3] {$(\{ s_0, s_2 \},3)$};
 \node[state]                     (5) [right=of 4] {$(\{ s_2 \},3)$};
 \path[->] (1) edge node {$a_2$} (2);
 \path[->] (2) edge node {$a_2$} (3);
 \path[->] (3) edge node {$a_2$} (4);
 \path[->] (4) edge node {$a_2$} (5);
 \path[->] (5) edge [loop right] node {$a_2$} (5);
 \node at (-1.3,0) {$\mathcal A_{(3,0)}^{(2)}$};
\end{tikzpicture}}
\caption{\footnotesize The reachable part of the unary automata $\mathcal A_{(0,0)}^{(2)}$,
     $\mathcal  A_{(1,0)}^{(2)}$, $\mathcal A_{(2,0)}^{(2)}$ and  $\mathcal A_{(3,0)}^{(2)}$ from Definition \ref{def:sequ_grid_decomp_aut},
     derived from the
     automaton from Example \ref{ex:grid_aut}. In Example \ref{ex:grid_aut}, these automata read in inputs in the up direction,
     but are drawn here horizontally to save space. See Example \ref{ex:seq_aut}
     for more explanation.}
    \label{fig::ex_sequ_aut}
\end{figure}
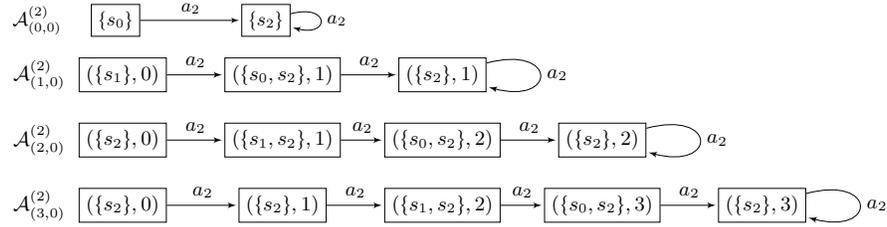
    \end{example}
    
    Suppose $p \in H_j$ and $j \in \{1,\ldots, k\}$.
    The next statement makes
    precise what we mean by decomposing the state label map along the 
    hyperplanes into the automata $\mathcal A_p^{(j)} = (\{a_j\}, Q_p^{(j)}, \delta_p^{(j)}, s_p^{(0,j)}, F_p^{(j)})$.
    Also, it justifies calling the first component of any state $(S, i) \in Q_p^{(j)}$
    also the state set label.
    
    \begin{restatable}[]{proposition}{gridautdecomp}{(state label map decomposition)} \label{prop:grid_aut_decomp}
     Suppose $\Sigma = \{a_1, \ldots, a_k\}$. 
     Let $1 \le j \le k$ and $p = (p_1, \ldots, p_k) \in \mathbb N_0^k$.
     Assume $\overline p \in H_j$ is the projection of $p$ onto $H_j$, i.e., $\overline p = (p_1, \ldots, p_{j-1}, 0, p_{j+1}, \ldots, p_k)$.
     Then
     $$
      \sigma_{\mathcal A}(p) = \pi_1(\delta_{\overline p}^{(j)}( s_{\overline p}^{(0,j)}, a_j^{p_j} ))
     $$
     for the automata $\mathcal A_{\overline p}^{(j)} = (\{a_j\}, Q_{\overline p}^{(j)}, \delta_{\overline p}^{(j)}, s_{\overline p}^{(0,j)}, F_{\overline p}^{(j)})$ from Definition \ref{def:sequ_grid_decomp_aut}.
    \end{restatable}
 
       With this observation, in Theorem \ref{thm:regularity_condition}, we derive a sufficient condition
       when the commutative image of some regular language is itself regular.
       It also gives us a general bound on the size of a minimal automaton, in case the commutative language
       is regular.
      
      \begin{restatable}[]{theorem}{regularitycondition} \label{thm:regularity_condition}
     Let $\mathcal A = (\Sigma, Q, \delta, s_0, F)$
     be a finite automaton. Suppose, for every $j \in \{1,\ldots, k\}$ and $p \in H_j$, with $H_j$ the hyperplane from Definition \ref{def:hyperplance}, the automata $\mathcal A_p^{(j)} = (\{a_j\}, Q_p^{(j)}, \delta_p^{(j)}, s_p^{(0,j)}, F_p^{(j)})$
     from Definition \ref{def:sequ_grid_decomp_aut}
     have a bounded number of states\footnote{Equivalently, the index and period is bounded, 
     which is equivalent 
     with just a finite number of distinct automata, up to (semi-automaton-)isomorphism.
     We call two automata (semi-automaton-)isomorphic if one automaton can be obtained from the other one by renaming states and alphabet symbols.}, i.e., $|Q_p^{(j)}| \le N$ for some $N \ge 0$ independent of $p$ and $j$.
     Then the commutative closure $\perm(L(\mathcal A))$     is regular
     and could be accepted by an automaton of size
     $$
      \prod_{j=1}^k (I_j + P_j),
     $$
     where $I_j$ denotes the largest index among
     the unary automata
     $
      \{ \mathcal A_p^{(j)} \mid p \in H_j \},
     $
     and $P_j$ the least common multiple
     of all the periods of these automata.
     In particular, by the relations of the index and period to the states from Section \ref{sec:unary},
     the automaton size is bounded by $N^k$.
     \todo{Muss ich das gleich verstehen? Intuition? 
     
     SH: Habe noch Erklärungen im folgenden hinzugefügt. Sowie eine Verbindung zur Ehrenfeucht/Haussler/Rozenberg-Arbeit.}
    \end{restatable} 
    
     This gives us a general bound in case the commutative closure
     is regular. We will apply this to the case of group languages and permutation automata
     in Section \ref{sec:grp_case}. Theorem \ref{thm:regularity_condition}
     has a close relation to Theorem 6.5 from \cite{EhrenfeuchtHR83}, namely
     case (iii), as we could link the periodic languages introduced in this paper
     to unary automata, as was done in \cite{DBLP:conf/cai/Hoffmann19,Hoffmann20}.
     This linkage, in general, allows us to give more concrete bounds and constructions.
     For example, we can list all periodic languages inside the commutative closure,
     or we can even give concrete bounds on resulting automata.
     The proof in \cite{EhrenfeuchtHR83} used more abstract well-quasi order arguments that do not yield
     concrete automata, nor do they allow the arguments we employ in Section \ref{sec:grp_case}.

    \subsection{The Special Case of Group Languages}
    \label{sec:grp_case}
    
     
     Here we apply Theorem \ref{thm:regularity_condition}
     to derive state bounds for group languages.
     We need some basic observations about permutations, see for example \cite{cameron_1999}.
     Every permutation 
     could be written in terms of disjoint cycles. 
     For an element\footnote{In this context, the elements are also called points in the literature, but
     we will stick to the term elements or states.}of the permutation domain, 
     by the \emph{cycle length} of that element with respect to a given permutation,
     we mean the length of the cycle in which this element appears\footnote{For a given element $m \in [n]$
     and a permutation $\pi : [n] \to [n]$, this is the number $|\{ \pi^i(m) \mid i \ge 0 \}|$, in the literature also called
     the \emph{orbit length} of $m$ under the subgroup generated by $\pi$.}. 
     The order of a permutation is the smallest power such that the identity permutation
     results, which equals the least common multiple of all cycle lengths for all elements.
      Before stating our results, let us make some general assumptions
      and fix some notions, to make the statements more concise.
      
      \begin{Assumption}{(general assumptions for this section)}\label{ass:grp_case}
       Let $\Sigma = \{a_1, \ldots, a_k\}$. 
       Assume  a permutation automaton
       $\mathcal A = (\Sigma, Q, \delta, s_0, F)$ is given
       with $j \in \{1,\ldots, k\}$ and a point $p \in H_j$, where
        $H_j$ denotes the hyperplane from Definition \ref{def:hyperplance}.
       We denote by $\mathcal A_p^{(j)} = (\{a_j\}, Q_p^{(j)}, \delta_p^{(j)}, s_p^{(0,j)}, F_p^{(j)})$ the automata
       from Definition \ref{def:sequ_grid_decomp_aut}.
       By $L_j$ we will denote the order
       of the letter $a_j$, viewed as a permutation of the states $Q$, i.e., the least common multiple  of the cycle lengths of all states.
      \end{Assumption}
     
      A crucial ingredient to our arguments will be the following observation. 
      
      \begin{restatable}[]{lemma}{stablelabelsgrowgrpcase} \label{lem:state_labels_grow_grp_case}
       Choose the notation from Assumption \ref{ass:grp_case}.
       Then the state set labels of states from $\mathcal A_p^{(j)}$ will not decrease in cardinality
       as we read in symbols, and their cardinality will stay constant on the cycle of $\mathcal A_p^{(i)}$.
       More precisely, let
       $(S, x), (T,y) \in Q_p^{(i)}$ be any states.
       If $(T,y) = \delta_p^{(i)}((S,x), a_i^r)$ for some $r \ge 0$,
       then $|T| \ge |S|$.
       And if $(S,x)$ and $(T,y)$ are both on the cycle, i.e.,
       $(S,x) = \delta_p^{(i)}((S,x), a_i^r)$
       and $(T,y) = \delta_p^{(i)}(S,x), a_i^s)$ for some $r,s \ge 0$,
       then $|S| = |T|$.
      \end{restatable}
      
      To give state bounds on a resulting automaton, using Theorem \ref{thm:regularity_condition},
      we need bounds on the indices and periods of the unary automata
      from Definition \ref{def:sequ_grid_decomp_aut}.
      The following result gives us a criterion when we have reached the cycle
      in these automata, and will be used in deriving the mentioned bounds.
      
      \begin{restatable}[]{lemma}{periodApgroup} \label{lem:period_Ap_group}
       Choose the notation from Assumption \ref{ass:grp_case}.
       Set\footnote{For $p \in H_j$, the condition $p = q + \psi(b)$, for some $b \in \Sigma$,
      implies $q \in H_j$ and $b \ne a_j$.}
       $$
         \mathcal P =\{ \mathcal A_q^{(j)} \mid p = q + \psi(b) \mbox{ for some } b \in \Sigma \}.
       $$
       Denote by $I$ the maximal index and by $P$ the least common multiple of the periods
       of the unary automata in $\mathcal P$.
       Suppose $S \subseteq Q$ and let $L_S = \lcm\{ | \{ \delta(s, a_j^i) : i \ge 0 \} | : s \in S\}$ be the least common multiple of the cycle lengths of the elements 
       in $S$ with respect to the letter $a_j$, seen as a permutation of the states.
       Then for $m \ge I$ and the states $(S,x), (T,y) \in Q_p^{(j)}$ which fulfill
       $$
        (S, x) = \delta_p^{(j)}( s_p^{(0,j)}, a_j^{m} ) \quad\mbox{and}\quad
        (T, y) = \delta_p^{(j)}( s_p^{(0,j)}, a_j^{m+\lcm(P, L_S)} )
       $$
       we have that if $|S| = |T|$, then $S = T$ and\footnote{As we assume $m \ge I$, by Equation \eqref{eqn:def_unary_aut_transition}
       from Definition \ref{def:sequ_grid_decomp_aut},
       we have $x \ge I$.} $x = y$.
       This also implies that the period of $\mathcal A_p^{(j)}$ divides $\lcm(P, L_S)$.
      \end{restatable}
 
       The next results gives us a bound for the periods of the automata 
       from Definition \ref{def:sequ_grid_decomp_aut}.
      
        \begin{restatable}[]{proposition}{boundperiodsApgroup} \label{prop:bound_periods_Ap_group}
         Choose the notation from Assumption \ref{ass:grp_case}. Let $p \in H_j$.
         Then the periods of all automata $\mathcal A_p^{(j)}$ divide $L_j$.
        \end{restatable}

        The criterion for the cycle detection from Lemma \ref{lem:period_Ap_group}
        could be a little bit relaxed by the next result, which will be
        more useful for proving a bound on the index of the automata from Definition \ref{def:sequ_grid_decomp_aut}.
        Intuitively, it bounds the way in which the indices of the automata
        from Definition \ref{def:sequ_grid_decomp_aut} can grow.

        \begin{restatable}[]{corollary}{Lioncycle} \label{cor:Li_on_cycle}
         Choose the notation from Assumption \ref{ass:grp_case}.
         Set
         $$
          \mathcal P =\{ \mathcal A_q^{(j)} \mid p = q + \psi(b) \mbox{ for some } b \in \Sigma \}.
         $$
         Denote by $I$ the maximal index and by $P$ the least common multiple of the periods
         of the unary automata in $\mathcal P$.
         Then for states $(S,x), (T,y) \in Q_p^{(j)}$ with $x \ge I$ and
         $$
          (T,y) = \delta_p^{(j)}( (S,x), a_i^{L_j} )
         $$ 
         we have that $|T| = |S|$ implies $T = S$ and $x = y$.
        \end{restatable}
      
        Finally, we state
        a bound for the indices of the automata from Definition \ref{def:sequ_grid_decomp_aut}.\todo{unvollst. Satz SH: Umformuliert.}
      
       \begin{restatable}[]{proposition}{indexspecificset} \label{prop:index_specific_set}
         Choose the notation from Assumption \ref{ass:grp_case}.
         Then the index of any automaton 
         $\mathcal A_p^{(j)}$ is 
         bounded by $(|T|-1) \cdot L_j$, where $T$ is any state set label from a state on the cycle of $\mathcal A_p^{(j)}$.
       \end{restatable}

      Combining everything gives our state complexity bound.
      
      \begin{theorem} 
       \label{thm:sc_group_case}
       Choose the notation from Assumption \ref{ass:grp_case}.
       Then the commutative closure $\perm(L(\mathcal A))$
       is regular and could be accepted by an automaton with at most
       \begin{equation}\label{eqn:sc_group_case} 
        \prod_{j=1}^k ( (|Q|-1)L_j + L_j ) 
         = |Q|^k \left( \prod_{j=1}^k L_j \right) 
       \end{equation}
       states.
      \end{theorem}
      \begin{proof} 
      
       First note that Proposition \ref{prop:index_specific_set} gives in particular
       that the indices of all automata are at most $(|Q|-1) L_j$. 
       %
       Also Proposition \ref{prop:bound_periods_Ap_group} yields the bound $L_j$
       for the periods.
       So Theorem \ref{thm:regularity_condition} gives the result. $\qed$
      \end{proof}
     
     
        \begin{example}\label{ex:perm_aut} \footnotesize
  Let $\Sigma = \{a_1, a_2\}$ and consider the permutation automaton from Figure~\ref{fig::ex_perm_aut}.
  It is the same automaton as given in~\cite{DBLP:journals/iandc/GomezGP13}. As an example
  for the group language case, we give its state labelling on $\mathbb N_0^k$
  and an automaton for the commutative closure, constructed from the unary automata $\mathcal A_p^{(j)}$. 
  Note that this is not the minimal automaton, which could be found in \cite{DBLP:journals/iandc/GomezGP13}. 
   Also, note that, with the notational convention from Assumption \ref{ass:grp_case},
   we have $L_1 = 3$ and $L_2 = 2$. Hence Theorem \ref{thm:sc_group_case}
   gives the bound $3^2 \cdot 6 = 54$. The automaton constructed from the unary automaton
   $\mathcal A_{p}^{(j)}$ is much smaller here, as the indices stabilize much faster than given by
   the theoretical bound.
  \begin{figure}[htb]
\begin{minipage}{0.3\textwidth}
\scalebox{.87}{
\begin{tikzpicture}[>=latex',shorten >=1pt,node distance=2cm,on grid,auto]
 \node[state, initial, accepting] (1) {$s_0$};
 \node[state]                     (3) [below right of=1] {$s_2$};
 \node[state]                     (2) [above right of=3] {$s_1$};
 \node at (1.2,-3.5) {$Q = \{s_0, s_1, s_2\}$};
 
 \path[->] (1) edge [bend left=20] node {$a_1, a_2$} (2);
 \path[->] (2) edge [bend left=20] node {$a_2$} (1);
 \path[->] (3) edge [left] node {$a_1$} (1);
 \path[->] (2) edge [right] node {$a_1$} (3);
 \path[->] (3) edge [loop below] node {$a_2$} (3);
  
\end{tikzpicture}}
\end{minipage}%
\begin{minipage}{0.7\textwidth}
\scalebox{.85}{
    \begin{tikzpicture}[yscale=1.1, xscale=1.45] 

        \node at (0,0) (0p0) {$\{s_0\}$};
        \node at (1,0) (1p0) {$\{s_1\}$};
        \node at (2,0) (2p0) {$\{s_2\}$};
        \node at (3,0) (3p0) {$\{s_0\}$};
        \node at (4,0) (4p0) {$\{s_1\}$};
        \node at (5,0) (5p0) {$\ldots$};
         
        \node at (0,1) (0p1) {$\{s_1\}$};
        \node at (0,2) (0p2) {$\{s_0\}$};
        \node at (0,3) (0p3) {$\{s_1\}$};
        \node at (0,4) (0p4) {$\{s_0\}$};
        \node at (0,5) (0p5) {$\vdots$};
        \node at (1,5) (1p5) {$\vdots$};
        \node at (2,5) (2p5) {$\vdots$};
        \node at (3,5) (3p5) {$\vdots$};
        \node at (4,5) (4p5) {$\vdots$};
        
        \node at (1,1) (1p1) {$\{s_0,s_2\}$};
        \node at (2,1) (2p1) {$Q$};
        \node at (3,1) (3p1) {$Q$};
        \node at (4,1) (4p1) {$Q$};
        \node at (5,1) (5p1) {$\ldots$};  
        
        \node at (1,2) (1p2) {$\{s_1,s_2\}$};
        \node at (2,2) (2p2) {$Q$};
        \node at (3,2) (3p2) {$Q$};
        \node at (4,2) (4p2) {$Q$};
        \node at (5,2) (5p2) {$\ldots$};  
        
        \node at (1,3) (1p3) {$\{s_0,s_2\}$};
        \node at (2,3) (2p3) {$Q$};
        \node at (3,3) (3p3) {$Q$};
        \node at (4,3) (4p3) {$Q$};
        \node at (5,3) (5p3) {$\ldots$}; 
        
        \node at (1,4) (1p4) {$\{s_1,s_2\}$};
        \node at (2,4) (2p4) {$Q$};
        \node at (3,4) (3p4) {$Q$};
        \node at (4,4) (4p4) {$Q$};
        \node at (5,4) (5p4) {$\ldots$};

       \path[->] (0p0) edge [below] node {$a_1$} (1p0)
                 (1p0) edge [below] node {$a_1$} (2p0)
                 (2p0) edge [below] node {$a_1$} (3p0)
                 (3p0) edge [below] node {$a_1$} (4p0)
                 (4p0) edge [below] node {$a_1$} (5p0);
                 
       \path[->] (0p0)  edge [left] node {$a_2$} (0p1)
                 (0p1)  edge [left] node {$a_2$} (0p2)
                 (0p2)  edge [left] node {$a_2$} (0p3)
                 (0p3)  edge [left] node {$a_2$} (0p4)
                 (0p4)  edge [left] node {$a_2$} (0p5);
                 
       \path[->] (0p1) edge [below] node {$a_1$} (1p1)
                 (1p1) edge [below] node {$a_1$} (2p1)
                 (2p1) edge [below] node {$a_1$} (3p1)
                 (3p1) edge [below] node {$a_1$} (4p1)
                 (4p1) edge [below] node {$a_1$} (5p1);
                 
       \path[->] (0p2) edge [below] node {$a_1$} (1p2)
                 (1p2) edge [below] node {$a_1$} (2p2)
                 (2p2) edge [below] node {$a_1$} (3p2)
                 (3p2) edge [below] node {$a_1$} (4p2)
                 (4p2) edge [below] node {$a_1$} (5p2);
                 
       \path[->] (0p3) edge [below] node {$a_1$} (1p3)
                 (1p3) edge [below] node {$a_1$} (2p3)
                 (2p3) edge [below] node {$a_1$} (3p3)
                 (3p3) edge [below] node {$a_1$} (4p3)
                 (4p3) edge [below] node {$a_1$} (5p3);
                 
       \path[->] (0p4) edge [below] node {$a_1$} (1p4)
                 (1p4) edge [below] node {$a_1$} (2p4)
                 (2p4) edge [below] node {$a_1$} (3p4)
                 (3p4) edge [below] node {$a_1$} (4p4)
                 (4p4) edge [below] node {$a_1$} (5p4);
                 
       \path[->] (1p0)  edge [left] node {$a_2$} (1p1)
                 (1p1)  edge [left] node {$a_2$} (1p2)
                 (1p2)  edge [left] node {$a_2$} (1p3)
                 (1p3)  edge [left] node {$a_2$} (1p4)
                 (1p4)  edge [left] node {$a_2$} (1p5);
                 
       \path[->] (2p0)  edge [left] node {$a_2$} (2p1)
                 (2p1)  edge [left] node {$a_2$} (2p2)
                 (2p2)  edge [left] node {$a_2$} (2p3)
                 (2p3)  edge [left] node {$a_2$} (2p4)
                 (2p4)  edge [left] node {$a_2$} (2p5);
                 
       \path[->] (3p0)  edge [left] node {$a_2$} (3p1)
                 (3p1)  edge [left] node {$a_2$} (3p2)
                 (3p2)  edge [left] node {$a_2$} (3p3)
                 (3p3)  edge [left] node {$a_2$} (3p4)
                 (3p4)  edge [left] node {$a_2$} (3p5);
                 
       \path[->] (4p0)  edge [left] node {$a_2$} (4p1)
                 (4p1)  edge [left] node {$a_2$} (4p2)
                 (4p2)  edge [left] node {$a_2$} (4p3)
                 (4p3)  edge [left] node {$a_2$} (4p4)
                 (4p4)  edge [left] node {$a_2$} (4p5);

    \node at (-1.3, -0.3) {$(I_1, P_1)$};             
    \node at (-0.7,-0.7) {$(I_2, P_2)$};
    \draw (-1.3,-0.6) -- (-0.3,-0.3);
    
    \node at (0, -0.5) {$(0, 2)$};     
    \node at (1, -0.5) {$(1, 2)$};   
    \node at (2, -0.5) {$(1, 2)$};   
    \node at (3, -0.5) {$(1, 2)$};   
    \node at (4, -0.5) {$(1, 2)$};   
     
    \node at (-0.7, 0) {$(0, 3)$};     
    \node at (-0.7, 1) {$(2, 3)$};   
    \node at (-0.7, 2) {$(2, 3)$};   
    \node at (-0.7, 3) {$(2, 3)$};   
    \node at (-0.7, 4) {$(2, 3)$};    
    \end{tikzpicture}}
\end{minipage}
\begin{center} 
\scalebox{.73}{ 
\begin{tikzpicture}[>=latex',shorten >=1pt,node distance=2cm,on grid,auto]
 \node[state, initial, accepting] (1) {};
 \node[state]                     (2) [right of=1] {};
 \node[state]           (3) [right of=2] {};
 \node[state,accepting]                     (4) [right of=3] {};
 \node[state]           (5) [right of=4] {};
 
 \node[state] (6) [below of=1] {};
 \node[state,accepting] (7) [right of=6] {};
 \node[state,accepting] (8) [right of=7] {};
 \node[state,accepting] (9) [right of=8] {};
 \node[state,accepting] (10) [right of=9] {};
 
 \node[state,accepting] (11) [below of=6] {};
 \node[state] (12) [right of=11] {};
 \node[state,accepting] (13) [right of=12] {};
 \node[state,accepting] (14) [right of=13] {};
 \node[state,accepting] (15) [right of=14] {};
 
 \path[->] (1) edge node {$a_1$} (2);
 \path[->] (2) edge node {$a_1$} (3);
 \path[->] (3) edge node {$a_1$} (4);
 \path[->] (4) edge node {$a_1$} (5);
 \path[->] (5) edge  [bend right=40] node [above] {$a_1$} (3);
 
 \path[->] (6) edge node {$a_1$} (7);
 \path[->] (7) edge node {$a_1$} (8);
 \path[->] (8) edge node {$a_1$} (9);
 \path[->] (9) edge node {$a_1$} (10);
 \path[->] (10) edge  [bend right=30] node [above] {$a_1$} (8);
 
 \path[->] (11) edge node {$a_1$} (12);
 \path[->] (12) edge node {$a_1$} (13);
 \path[->] (13) edge node {$a_1$} (14);
 \path[->] (14) edge node {$a_1$} (15);
 \path[->] (15) edge  [bend left=40] node [above] {$a_1$} (13);
 
 \path[->] (1)  edge node {$a_2$} (6);
 \path[->] (6)  edge [bend left=20] node {$a_2$} (11);
 \path[->] (11) edge [bend left=20] node {$a_2$} (6);
 
 \path[->] (2)  edge node {$a_2$} (7);
 \path[->] (7)  edge [bend left=20] node {$a_2$} (12);
 \path[->] (12) edge [bend left=20] node {$a_2$} (7);
 
 \path[->] (3)  edge node {$a_2$} (8);
 \path[->] (8)  edge [bend left=20] node {$a_2$} (13);
 \path[->] (13) edge [bend left=20] node {$a_2$} (8);
 
 \path[->] (4)  edge [bend left=30] node {$a_2$} (9);
 \path[->] (9)  edge [bend left=20] node {$a_2$} (14);
 \path[->] (14) edge [bend left=20] node {$a_2$} (9);
 
 \path[->] (5)  edge node {$a_2$} (10);
 \path[->] (10) edge [bend left=20] node {$a_2$} (15);
 \path[->] (15) edge [bend left=20] node {$a_2$} (10);
\end{tikzpicture}}
\end{center}
 \caption{\footnotesize The constructions from this paper for a permutation automaton,
 different from the one given in Figure \ref{fig::example}. In the state labelling
 of $\mathbb N_0^{|\Sigma|}$, the origin is in the bottom left corner, labeled by $\{s_0\}$.
 Also indicated, written beneath, or to the side of, the axes, are the indices and periods
 of the unary automata in the direction of $\varphi(a_j)$ from Definition \ref{def:sequ_grid_decomp_aut}.
 See Example \ref{ex:perm_aut} for explanations.}
    \label{fig::ex_perm_aut}
\end{figure}
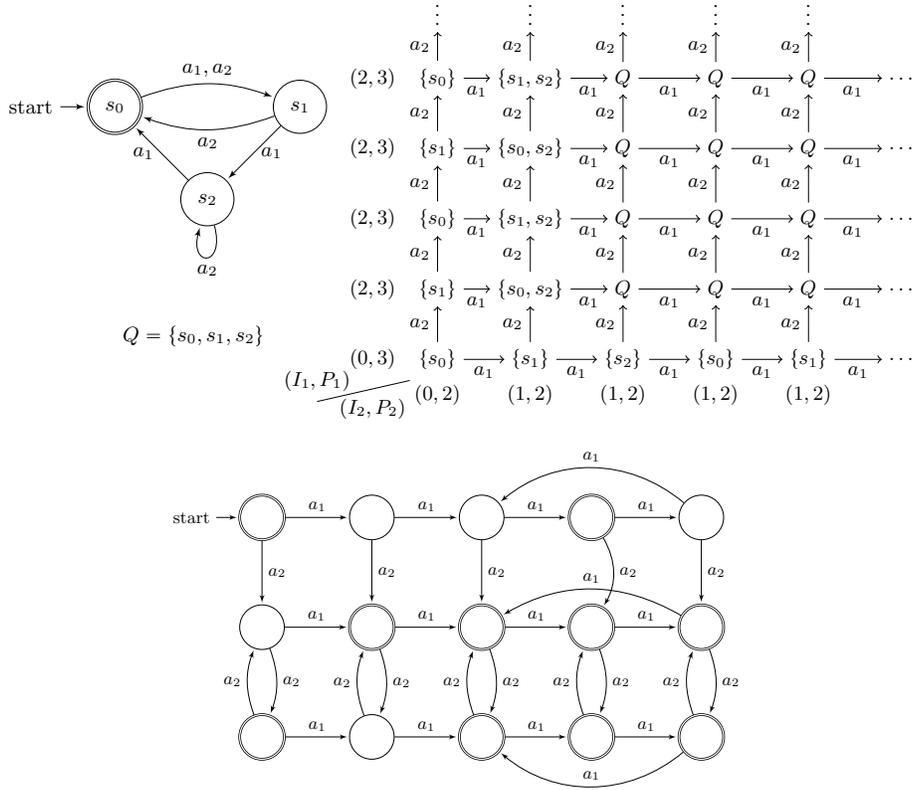
\end{example}
    
       \begin{example} \footnotesize
        Let $\Sigma = \{a_1, a_2\}$ and consider 
        $\mathcal A = (\Sigma, Q, \delta, s_0, F)$ with $Q = [n]$
        for some $n \ge 1$, $s_0$ and $F$ arbitrary, and $\delta(0, a_1) = 1, \delta(1, a_1) = 0, \delta(x,a_1) = x$
        for $x \in \{2,\ldots,n-1\}$, $\delta(x, a_2) = (x+1) \bmod n$ for $x\in [n]$.
        Then $\perm(L(\mathcal A))$ could be accepted by an automaton
        of size $2n^3$.
       \end{example}
       
      
       As stated in \cite{GaoMRY17}, the maximal order of any permutation
       on a set of size $n$ is given by Landau's function, which 
       is asymptotically like $e^{\Theta(\sqrt{n\ln n})}$. 
       Hence, Theorem~\ref{thm:state_compl_grp_lang},
       the asymptotic form of Theorem~\ref{thm:sc_group_case}, is implied.

\section{Conclusion}





We have shown that the commutative closure of regular group languages is regular, 
and have derived a bound on the size of the resulting automaton. 
The size
is related to the least common multiples of the cycle lengths of the letters,
viewed as permutations on the states, see Equation \eqref{eqn:sc_group_case}.
I do not know if the bound is sharp. I have not found a single example
that has the property
that the index of the constructed automata $\mathcal A_p^{(j)}$,
for the letter $a_j$, has length $(|Q|-1)\cdot L_j$,
as would be necessary to reach the bound stated in Theorem \ref{thm:sc_group_case}.
In fact, I believe that the cycles on individual elements of the state label are never ``traversed'' in its entirety
before another element is added to the state label, or we reach the final cycle of the unary automata $\mathcal A_p^{(j)}$.
So I conjecture that for larger alphabets we can improve this bound, as the state
labels grow faster in the index part of the automata $\mathcal A_p^{(j)}$,
as more predecessor automata\footnote{For 
some automaton $\mathcal A_p^{(j)}$, with $p \in H_j$ and $j \in \{1,\ldots, k\}$,
all automata $\mathcal A_q^{(j)}$ with $p = q + \psi(b)$, for some $b \in \Sigma$,
are called \emph{predecessor automata} of $\mathcal A_p^{(j)}$.}
add states of the original automaton to the state labels of $\mathcal A_p^{(j)}$, as
inputs are read.
This is somehow contrary to what usually happens in other existing state complexity
results, namely that we need larger alphabets to reach the state bounds, see for example 
\cite{DBLP:journals/ijfcs/HanS08,DBLP:conf/dlt/HanS07,BrzozowskiJLRS16}.
In our situation, I somehow conjecture that for larger alphabets (where surely, distinct letters have to be
distinct permutations), indices
of the unary automata $\mathcal A_p^{(j)}$ get smaller and smaller. Hence
the overall state complexity bound reaches the product of the least common multiples
of the cycle lengths for all letters, i.e., we have $\prod_{j=1}^k L_j$ as a bound in the limit for $k \to \infty$,
with an alphabet of size $k$.


\smallskip \noindent \footnotesize
\textbf{Acknowledgement.} I thank my supervisor, Prof. Dr. Henning Fernau, for giving valuable feedback and remarks on the content of this article that improved its presentation.

\todo[inline]{Wieso gibt es hier unsere JALC Arbeit doppelt? Einmal mit, einmal ohne DOI. Du solltest auch mal auf die 16 Fehlermeldungen achten, die OverLeaf anzeigt. Besser sind 0 Fehler.

SH: Berichtigt. Aber zwei Warnungen bekomme ich nicht weg. Und diese
\texttt{overfull hhbox}-Warnungen find ich auch immer recht mysteriös.}

\bibliographystyle{splncs04}
\bibliography{ms}

\clearpage
\section{Appendix}
Here, we collect some proofs not given in the main text. 
Our first Lemma \ref{lem:ind_form_state_label} is not stated in the main text, as it is essentially
only used in the proofs of this appendix.

\begin{lemma}{(inductive form of state label function)}\label{lem:ind_form_state_label}
 Let $\mathcal A = (\Sigma, Q, \delta, s_0, F)$ be some finite automaton
 and $\sigma_{\mathcal A} : \mathbb N_0^k \to \mathcal P(Q)$ 
 the state label function from Definition \ref{def:state-label-function}.
 Then we have $\sigma_{\mathcal A}(0,\ldots, 0) =\{s_0\}$, and
 \begin{equation}\label{eqn:state_label_fn_ind_form}
  \sigma_{\mathcal A}(p) =  \bigcup_{ \substack{ (q,b) \\ p = q + \psi(b) }} \delta( \sigma_{\mathcal A}(q), b ) 
 \end{equation} 
 for $p \ne (0,\ldots 0)$.
\end{lemma}
\begin{proof}
  If $p = (0,\ldots, 0)$, then $\{ \delta(s_0, w) : \psi(w) = p \} = \{\delta(s_0, \varepsilon)\} = \{s_0\}$.
  Suppose $p \ne (0,\ldots, 0)$.
  Then 
  \begin{align*}  
     & \bigcup_{ \substack{ (q,b) \\ p = q + \psi(b) }} \delta( \sigma_{\mathcal A}(q), b ) \\
     & = \bigcup_{ \substack{ (q,b) \\ p = q + \psi(b) }} \delta( \{ \delta(s_0, w) \mid \exists w \in \Sigma^* : \psi(w) = q \}, b) \\ 
     & = \bigcup_{ \substack{ (q,b) \\ p = q + \psi(b) }} \{ \delta(s_0, wb) \mid \exists w \in \Sigma^* : \psi(w) = q \} \\ 
     & = \bigcup_{ \substack{ (q,b) \\ p = q + \psi(b) }} \{ \delta(s_0, wb) \mid \exists w \in \Sigma^* : \psi(w) + \psi(b) = p, \psi(w) = q \} \\ 
     & = \{ \delta(s_0, wb) \mid \exists w \in \Sigma^* \exists b \in \Sigma : \psi(wb) = p \} \\ 
     & = \{ \delta(s_0, u) \mid \exists u \in \Sigma^* : \psi(u) = p \}. \quad\qed
  \end{align*}
\end{proof}

We also need the next Lemma here.

\begin{restatable}[]{lemma}{unaryperioddivides}\label{lem:unary_period_divides}
 Let $\mathcal A = (\Sigma, Q, \delta, s_0, F)$ be some unary automaton.
 If $\delta(s, a^k) = s$ for some state $s \in Q$ 
 and number $k > 0$, then $k$ is divided by the period of $\mathcal A$.
\end{restatable}
\begin{proof}
 Let $i$ be the index, and $p$ the period of $\mathcal A$.
 We write $k = np + r$ with $0 \le r < p$.
 First note that $s$ is on the cycle of $\mathcal A$, i.e.,
 $$
  s \in \{ \delta(s_0, a^i), \delta(s_0, a^{i+1}),\ldots, \delta(s_0,a^{i+p-1}) \}
 $$
 as otherwise $i$ would not be minimal. Then if $s = \delta(s_0, a^{i+j})$ for some $0 \le j < p$
 we have $\delta(s_0, a^{i+k}) = 
 \delta(s_0, a^{i+p+k}) = \delta(s_0, a^{i+j+k+(p-j)}) = \delta(s_0,a^{i+j+(p-j)}) = \delta(s_0,a^i)$.
 So $\delta(s_0,a^i) = \delta(s_0,a^{i+k}) = \delta(s_0, a^{i+np + r}) = \delta(s_0, a^{i+r})$
 which gives $r = 0$ by minimality of $p$. $\qed$
\end{proof}

     The next simple observation will be used. 
     
     \begin{lemma} \label{lem:lcm_cycle_subset} 
      Let $A \subseteq [n]$. 
      Then if $m$ is the least common multiple of  all the different cycles lengths of elements
      from $A$
      for a permutation $\pi : [n] \to [n]$, then $\pi^m_{|A} = \operatorname{id}_{|A}$,
      in particular $\pi^m(A) = A$. 
     \end{lemma} 
     

\subsection{Proof of Proposition~\ref{prop:conn_parikh} (See page~\pageref{prop:conn_parikh})}  
\connparikh*
 \begin{proof} 
   Notation as in the statement of the Proposition.
   For $p \in \mathbb N_0^k$, we have
   \begin{align*}
      \sigma_{\mathcal A}(p) \cap F \ne \emptyset 
       & \Leftrightarrow \{ \delta(s_0, w) \mid \exists w \in \Sigma^* : \psi(w) = p \} \cap F \ne \emptyset \\ 
       & \Leftrightarrow \exists w \in \Sigma^* : \delta(s_0, w) \in F \mbox{ and } \psi(w) = p \\
       & \Leftrightarrow \exists w \in \Sigma^* : w \in L(\mathcal A) \mbox{ and } \psi(w) = p \\
       & \Leftrightarrow p \in \psi(L(\mathcal A)) \quad \qed
   \end{align*}
\end{proof}

\begin{remark}{(induction scheme used)} \label{rem:inductive_form} 
In certain proofs, namely of Proposition \ref{prop:grid_aut_decomp}, Proposition \ref{prop:bound_periods_Ap_group}
and Proposition \ref{prop:index_specific_set}, we argue in an inductive fashion. Also, the
formulation of Lemma \ref{lem:ind_form_state_label} is inductive. This comes from 
the inductive form that the automata from Definition \ref{def:sequ_grid_decomp_aut}
are defined, or the recursive way that they are created from previous automata.
Just in case you are wondering why, in the inductive proofs
of Proposition \ref{prop:bound_periods_Ap_group} and Proposition \ref{prop:index_specific_set}, no base case is explicitly stated,
I will give some justification for that in the next paragraph. But in case you are not wondering, you might well skip
this explanation. Suppose we have some property $A$
that we want to show is true for all automata
$\mathcal A_p^{(j)} = (\{a_j\}, Q_p^{(j)}, \delta_p^{(j)}, s_p^{(0,j)}, F_p^{(j)})$
from Definition \ref{def:sequ_grid_decomp_aut}, where $j \in \{1,\ldots, k\}$
and $p \in H_j$, the hyperplane
         from Definition \ref{def:hyperplance}.
Then, our induction scheme is the following.

\begin{quote}
    Fix some $j \in \{1,\ldots, k\}$.
    If we can show property $A$ for $\mathcal A_p^{(j)}$
    under the assumption that it is true for all automata
    from the set
    $$
     \mathcal P = \{ \mathcal A_q^{(j)} \mid  p = q + \psi(b) \mbox{ for some } b \in \Sigma \},
    $$
    then it is true for all automata $\mathcal A_p^{(j)}$,
    for $p \in H_j$ arbitrary.
\end{quote}

In all our cases, the base case is when $\mathcal P = \emptyset$, and our arguments
will work in that case too. Hence, there is no need to treat that as a special (induction base)
case. More specifically, we will use the maximal index and the least common multiple
of the automata from $\mathcal P$. As $\max\emptyset = 0$ and $\lcm \emptyset = 1$, by definition,
the arguments, given below in the proofs, will work with these values. Even the original Definition \ref{def:sequ_grid_decomp_aut}
has no explicit base case, but relies on these definitions\footnote{Note
that $\mathcal P = \emptyset$ if and only $p = (0,\ldots,0)$, in which
case $\mathcal A_p^{(j)}$ is isomorphic to the starting automaton $\mathcal A$.}.

This is related to the fact that $\mathbb N_0^k$ is well-quasi order (or a well partial order to be
more specific). So also the points from $H_j$ are well-quasi ordered. 
As $\psi(a_j) = (0, \ldots, 0, 1, 0, \ldots, 0)$, where the one appears precisely at the $j$-th position,
the condition $p = q + \psi(a_j)$ says that $q$ is an immediate predecessor point.
The induction scheme we use hence reduces to an induction scheme
over this well partial order. A justification of this induction principle
for well-quasi orders could be found, for example, in the thesis \emph{On Well-Quasi-Orderings}\footnote{Thurman,~F.~B. \emph{On Well-Quasi-Orderings}. Thesis, University of Central Florida, Orlando, Florida (2013), \url{http://etd.fcla.edu/CF/CFH0004455/Thurman_Forrest_B_201304_BS.pdf}}, by Forrest B. Thurman.
\end{remark}

\subsection{Proof of Proposition~\ref{prop:grid_aut_decomp} (See page~\pageref{prop:grid_aut_decomp})} 
     \gridautdecomp*
    \begin{proof} 
      Notation as in the statement. For $p = (0, \ldots, 0)$ this is clear.
      If $p_j = 0$, then $p = \overline p$, and, by Equation \eqref{eqn:def_unary_aut_start_state},
      $$
       \pi_1(\delta_{\overline p}^{(j)}(s_{\overline p}^{(0,j)}, \varepsilon))
        = \pi_1( s_{\overline p}^{(0,j)} ) = \sigma_{\mathcal A}(\overline p).
      $$
      Suppose $p_j > 0$ from now on.
      Then, the set $\{ (q,b) \in \mathbb N_0^k \times \Sigma \mid p = q + \psi(b) \}$
      is non-empty, and we can use Equation \eqref{eqn:state_label_fn_ind_form},
      and proceed inductively
      \begin{align}
         \sigma_{\mathcal A}(p) & =  \bigcup_{ \substack{ (q,b) \\ p = q + \psi(b) }} \delta( \sigma_{\mathcal A}(q), b ) \nonumber \\ 
                              & =  \bigcup_{ \substack{ (q,b) \\ p = q + \psi(b) }} \delta( \pi_1(\delta_{\overline q}^{(j)}( s_{\overline q}^{(0,j)}, a_j^{q_j} )), b ) \label{eqn:ind_hyp_decomp_proof}
      \end{align}
      where $q = (q_1, \ldots, q_k)$, and $\overline q = (q_1, \ldots, q_{j-1},0,q_{j+1}, \ldots, q_k) \in H_j$.
      As $p_j > 0$ we have $p = q + \psi(a_j)$ for some unique point $q = (p_1, \ldots, p_{j-1}, p_j - 1, p_{j+1} \ldots, p_k)$. 
      For all other points $r = (r_1, \ldots, r_k)$ with $p = r + \psi(b)$ for some $b\in \Sigma$,
      the condition $r \ne q$ implies $b \ne a_j$ and $r_j = p_j$
      for $r = (r_1, \ldots, r_k)$. Also, if $\overline q \in H_j$ denotes projection to $H_j$, we have $\overline q  = \overline p$
      for our chosen $q$ with $p = q + \psi(a_j)$.
      Hence, taken all this together, we can write Equation \eqref{eqn:ind_hyp_decomp_proof}
      in the form 
      \begin{equation*}
       \sigma_{\mathcal A}(p) = \left( 
          \bigcup_{ \substack{ (r,b), b \ne a_j \\ p = r + \psi(b) }}
          \delta( \pi_1(\delta_{\overline r}^{(j)}( s_{\overline r}^{(0,j)}, a_j^{p_j} )), b ) \right)
          \cup 
          \delta( \pi_1(\delta_{\overline p}^{(j)}( s_{\overline p}^{(0,j)}, a_j^{p_j-1})), a_j).
      \end{equation*}
      Let $b \in\Sigma$. As for $a_j \ne b$, we have that $p = r + \psi(b)$ if and only if $\overline p = \overline r + \psi(b)$,
      with the notation as above for $p, r, \overline p$ and $\overline r = (r_1, \ldots, r_{j-1},0,r_{j+1}, \ldots, r_k)$, we can simplify further and write
      \begin{equation}\label{eqn:decomp_proof_psi_of_p}
       \sigma_{\mathcal A}(p) = \left( 
          \bigcup_{ \substack{ (\overline r,b), \overline r \in H_j \\ \overline p = \overline r + \psi(b) }}
          \delta( \pi_1(\delta_{\overline r}^{(j)}( s_{\overline r}^{(0,j)}, a_j^{p_j} )), b )  \right)
          \cup 
          \delta( \pi_1(\delta_{\overline p}^{(j)}( s_{\overline p}^{(0,j)}, a_j^{p_j-1})), a_j).
      \end{equation}

      \noindent Set $S = \pi_1(\delta_{\overline p}^{(j)}( s_{\overline p}^{(0,j)}, a_j^{p_j-1}))$,
      $T =  \sigma_{\mathcal A}(p)$ and\footnote{Note that for $\overline p \in H_j$, the condition $\overline p = q + \psi(b)$, for some $b \in \Sigma$,
      implies $q \in H_j$ and $b \ne a_j$.}
      $$
       \mathcal P = \{ \mathcal A_r^{(j)} \mid \overline p = r + \psi(b) \mbox{ for some } b \in \Sigma \}. 
      $$
      Let $I$ be the maximal index, and $P$ the least common multiple of all the periods, of
      the unary automata in $\mathcal P$. 
      We distinguish two cases for the value of $p_j > 0$.
      
      \begin{enumerate} 
      \item[(i)] $0 < p_j \le I$.
        
        By Equation \eqref{eqn:def_unary_aut_transition}, $\delta_{\overline p}( s_{\overline p}^{(0,j)}, a_j^{p_j-1} )= (S, p_j-1)$.
        In this case Equation \eqref{eqn:decomp_proof_psi_of_p}
        equals Equation \eqref{eqn:def_seq_grid_aut_transition}, if the state $(S, p_j-1)$ is used
        in Equation \eqref{eqn:def_unary_aut_transition}.
        This
        gives 
        $$
         \delta_{\overline p}^{(j)}( (S, p_j-1), a_j ) = (T, p_j).
        $$
        Hence $\pi_1(\delta_{\overline p}^{(j)}( (S, p_j-1), a_j )) = T = \sigma_{\mathcal A}(p)$.
      
      \item[(ii)] $I < p_j$.
      
        Set $y = I + ((p_j - 1 - I) \bmod P)$.
        Then $I \le y < I + P$.
        By Equation \eqref{eqn:def_unary_aut_transition},
        $
           \delta_{\overline p}^{(j)}( s_{\overline p}^{(0,j)}, a_j^{p_j-1} ) = ( S, y ).
        $
        So, also by Equation \eqref{eqn:def_unary_aut_transition},
        $$ 
          \delta_{\overline p}^{(j)}( s_{\overline p}^{(0,j)}, a_j^{p_j} ) = 
          \delta_{\overline p}^{(j)}( (S,y), a_j ) = 
          \left\{ 
           \begin{array}{ll}
            (R, y + 1) & \mbox{ if } I \le y < I + P - 1 \\ 
            (R, I)     & \mbox{ if } y = I + P - 1,
           \end{array}
          \right.
        $$
        where, by Equation \eqref{eqn:def_seq_grid_aut_transition},
        \begin{equation}\label{eqn:R}
         R = \delta(S, a_j) \cup \bigcup_{\substack{(\overline r, b) \\ \overline p = \overline r + \psi(b)}} \delta(\pi_1(\delta_{\overline r}^{(j)}(s_{\overline r}^{(0,j)}, a_j^{y+1})), b).
        \end{equation}
        Let $\overline r \in H_j$ with $\overline p = \overline r + \psi(b)$
        for some $b \in \Sigma$, and $\overline p \in H_j$ the point from the statement of this Proposition.
        Then, as the period of $\mathcal A_{\overline r}^{(j)}$ divides $P$, and $y$ is greater than or equal to
        the index of $\mathcal A_{\overline r}^{(j)}$, we have
        $$
          \delta_{\overline r}^{(j)}(s_{\overline r}^{(0,j)}, a_j^{p_j-1}) = \delta_{\overline r}^{(j)}(s_{\overline r}^{(0,j)}, a_j^{y}).
        $$
        So $\delta_{\overline r}^{(j)}(s_{\overline r}^{(0,j)}, a_j^{p_j}) = \delta_{\overline r}^{(j)}(s_{\overline r}^{(0,j)}, a_j^{y+1})$.
        Hence, comparing Equation \eqref{eqn:R} with Equation \eqref{eqn:decomp_proof_psi_of_p},
        we find $R = T$. $\qed$

        \medskip

        Alternative argument for case (ii): Set $x = I + ( (p_j - I) \bmod P )$. Then $I \le x < I + P$. By Equation \eqref{eqn:def_unary_aut_transition},
        $$
         \delta_{\overline p}^{(j)}( s_{\overline p}^{(0,j)}, a_j^{p_j-1} )
          = \left\{ \begin{array}{ll}
           (S, x - 1) & \mbox{ if } x > I \\ 
           (S, I + P - 1) & \mbox{ otherwise, if } x = I.
          \end{array}\right.
        $$
        We have
        $
         \delta_{\overline p}^{(j)}( s_{\overline p}^{(0,j)}, a_j^{p_j} )
           = \delta_{\overline p}^{(j)}(\delta_{\overline p}^{(j)}( s_{\overline p}^{(0,j)}, a_j^{p_j-1} ), a_j).
        $
        Hence, if $I < x < I + P$ we get
        $$
         \delta_{\overline p}^{(j)}( s_{\overline p}^{(0,j)}, a_j^{p_j} ) 
          = \delta_{\overline p}^{(j)}( (S, x - 1), a_j) = (R, x),
        $$
        with, by Equation \eqref{eqn:def_seq_grid_aut_transition},
        $$
         R = \delta(S, a_j)
           \cup \bigcup_{\substack{(\overline q, b) \\ \overline p = \overline q + \psi(b)}} \delta(\pi_1(\delta_q^{(j)}(s_q^{(0,j)}, a_j^{x})), b), 
        $$
        and if $x = I$, we get
        $$
         \delta_{\overline p}^{(j)}( s_{\overline p}^{(0,j)}, a_j^{p_j} ) 
          = \delta_{\overline p}^{(j)}( (S, I + P - 1), a_j) = (R, x),
        $$
        with, by Equation \eqref{eqn:def_seq_grid_aut_transition},
        $$
         R = \delta(S, a_j)
           \cup \bigcup_{\substack{(\overline q, b) \\ \overline p = \overline q + \psi(b)}} \delta(\pi_1(\delta_q^{(j)}(s_q^{(0,j)}, a_j^{I + P})), b).
        $$
        
        Let $\overline r \in H_j$ with $\overline p = \overline r + \psi(b)$
        for some $b \in \Sigma$, and $\overline p \in H_j$ the point from the statement of this Proposition.
        Then, as $P$ divides the period of $\mathcal A_{\overline r}^{(j)}$, and $x$ is greater than or equal to
        the index of $\mathcal A_{\overline r}^{(j)}$, we have
        $$
         \delta_{\overline r}^{(j)}( s_{\overline r}^{(0,j)}, a_j^{p_j} )
          = \delta_{\overline r}^{(j)}( s_{\overline r}^{(0,j)}, a_j^{x} ).
        $$
        Similar, 
        $\delta_{\overline r}^{(j)}( s_{\overline r}^{(0,j)}, a_j^{I + P} ) = \delta_{\overline r}^{(j)}( s_{\overline r}^{(0,j)}, a_j^{I + P} )$.
        So, combining everything so far, we have 
        $$
          \delta_{\overline p}^{(j)}( s_{\overline p}^{(0,j)}, a_j^{p_j} ) = (R, x)
        $$
        with
        $$
         R =  \delta(S, a_j)
           \cup \bigcup_{\substack{(\overline q, b) \\ \overline p = \overline q + \psi(b)}} \delta(\pi_1(\delta_q^{(j)}(s_q^{(0,j)}, a_j^{x})), b) = T.
        $$
        Hence $\sigma_{\mathcal A}(p) = T = R = \pi_1( \delta_{\overline p}^{(j)}( s_{\overline p}^{(0,j)}, a_j^{p_j} )$. $\qed$
     \end{enumerate}
      \end{proof}

\subsection{Proof of Theorem~\ref{thm:regularity_condition} (See page~\pageref{thm:regularity_condition})}       
\regularitycondition*
    \begin{proof}
     We use the same notation as introduced in the statement of the theorem.
     Let $p = (p_1, \ldots, p_k) \in \mathbb N_0^k$ and $j \in \{1,\ldots k\}$.
     Denote by $\sigma_{\mathcal A} : \mathbb N_0^k \to \mathcal P(Q)$ the state label function
     from Definition \ref{def:state-label-function}.
     Then, with Proposition \ref{prop:grid_aut_decomp},
     if $p_j \ge I_j$, we have 
     \begin{equation}\label{eqn:state_label_fn_ult_periodic}
     \sigma_{\mathcal A}(p_1, \ldots, p_{j-1}, p_j + P_j, p_{j+1}, \ldots, p_k)
      = \sigma_{\mathcal A}(p_1, \ldots, p_k).
     \end{equation}
     Construct the unary semi-automaton\footnote{The term semi-automaton is used
     for automata without a designated initial state, nor a set of final states.}
     $\mathcal A_j = (\{a_j\}, Q_j, \delta_j)$
     with 
     \begin{align*}
        Q_j & = \{s^{(j)}_0, s_1^{(j)}, \ldots, s_{I_j + P_j - 1}^{(j)} \}, \\
        \delta_j( s_i^{(j)} , a_j ) & = \left\{ \begin{array}{ll}
       s_{i+1}^{(j)} & \mbox{ if } i < I_j \\ 
       s_{I_j + (i-I_j+1) \bmod P_j}^{(j)} & \mbox{ if } i \ge I_j. \end{array}\right.
     \end{align*}
     Then build $\mathcal C = (\Sigma, Q_1 \times \ldots \times Q_k, \mu, s_0, E)$
     with
     \begin{align*}
       s_0                               & = (s_0^{(1)}, \ldots, s_0^{(k)}), \\
       \mu( (t_1, \ldots, t_k), a_i ) & = (t_1, \ldots, t_{j-1}, \delta_j(t_j, a_j), t_{j+1}, \ldots, t_k) \quad \mbox{ for all } 1 \le j \le k, \\
       E & = \{ \mu( s_0, u) : u \in L(\mathcal A) \} 
     \end{align*}
     By construction, for words $u,v \in \Sigma$ with $u \in \perm(v)$
     we have $\mu((t_1, \ldots, t_k), u) = \mu((t_1, \ldots, t_k), v)$
     for any state $(t_1, \ldots, t_k) \in Q_1 \times \ldots Q_k$.
     Hence, the language accepted by $\mathcal C$ is commutative.
     We will show that $L(\mathcal C) = \perm(L(\mathcal A))$.
     By choice of $E$ we have $L(\mathcal A) \subseteq L(\mathcal C)$,
     this gives $\perm(L(\mathcal A)) \subseteq L(\mathcal C)$.
     Conversely, suppose $w \in L(\mathcal C)$.
     Then $\mu(s_0, w) = \mu(s_0, u)$ for some $u \in L(\mathcal A)$.
     Next, we will argue that we can find $w' \in L(\mathcal C)$
     and $u' \in L(\mathcal A)$ with $\mu(s_0, w') = \mu(s_0, w) = \mu(s_0, u) = \mu(s_0, u')$ and $\max\{ |w'|_{a_j}, |u'|_{a_j} \} < I_j + P_j$
      for all $j \in \{1,\ldots, k\}$.
      
     \begin{enumerate}
     \item[(i)] By construction of $\mathcal C$, if $|w|_{a_j} \ge I_j + P_j$, we
      can find $w'$ with $|w'|_{a_j} = |w|_{a_j} - P_j$ such that
      $\mu(s_0, w') = \mu(s_0, w)$. So, inductively, suppose we have $w' \in \Sigma^*$
      with $|w'|_{a_j} < I_j + P_j$ for all $j \in \{1,\ldots, k\}$
      and $\mu(s_0, w) = \mu(s_0, w')$.
      
     \item[(ii)] By Corollary \ref{cor:state_label_described_com_closure}, $\sigma_{\mathcal A}(\psi(u))\cap F \ne \emptyset$.
      So, if $|u|_{a_j} \ge I_j + P_j$, by Equation \eqref{eqn:state_label_fn_ult_periodic},
      we can find $u'$ with $|u'|_{a_j} = |u|_{a_j} - P_j$
      and $\sigma_{\mathcal A}(\psi(u')) \cap F \ne \emptyset$.
      By definition of the state label function, as some word must induce the final state in $\sigma_{\mathcal A}(\psi(u'))$,
      we can choose $u' \in L(\mathcal A)$.
      Also, by construction of $\mathcal C$, we have $\mu(s_0, u) = \mu(s_0, u')$.
      So, after repeatedly applying the above steps, suppose we have $u' \in L(\mathcal A)$
      with $\mu(s_0,u) = \mu(s_0, u')$ and $|u'|_{a_j} < I_j + P_j$ for all $j \in \{1,\ldots, k\}$.
     \end{enumerate}
     By construction of $\mathcal C$, for words $u, v \in \Sigma^*$
     with $\max\{|u|_{a_j}, |v|_{a_j}\} < I_j + P_j$
     for all $j \in \{1,\ldots, k\}$ we have 
     \begin{equation}\label{eqn:C} 
      \mu(s_0, u) = \mu(s_0, v) \Leftrightarrow u \in \perm(v) \Leftrightarrow \psi(u) = \psi(v).
     \end{equation}
     Hence, using Equation \eqref{eqn:C} for the words $w'$ and $u'$ from (i) and (ii) above, as $\mu(s_0, u') = \mu(s_0, w')$, we
     find $\psi(u') = \psi(w')$.
     So $\sigma_{\mathcal A}(\psi(w')) \cap F \ne \emptyset$ as $u' \in L(\mathcal A)$.
     Now, again using Equation \eqref{eqn:state_label_fn_ult_periodic},
     this gives $\sigma_{\mathcal A}(\psi(w)) \cap F \ne \emptyset$,
     which, by Corollary \ref{cor:state_label_described_com_closure}, yields $w \in \perm(L(\mathcal A))$. $\qed$
    \end{proof}
    
\subsection{Proof of Lemma~\ref{lem:state_labels_grow_grp_case} (See page~\pageref{lem:state_labels_grow_grp_case})} 
\stablelabelsgrowgrpcase*
      \begin{proof}
       Notation as in the statement of the Lemma.
       By Equation \eqref{eqn:def_unary_aut_transition}
       and Equation \eqref{eqn:def_seq_grid_aut_transition}
       from Definition \ref{def:sequ_grid_decomp_aut}, intuitively,
       as we read in symbols
       in the automaton $\mathcal A_p^{(j)}$,
       the state set label of the next state is composed by the transition 
       from the previous state label, and by adding states from nearby automata. 
       And as we only apply permutations, those sets cannot get smaller.
       More formally, from Equation \eqref{eqn:def_seq_grid_aut_transition} of Definition \ref{def:sequ_grid_decomp_aut},
       we have that $\delta(S, a_j^i) \subseteq \pi_1(\delta_p^{(j)}((S,x), a_j^i))$
       for all $i \ge 0$.
       As $a_i$ induces a permutation on the state set, we have $|S| = |\delta(S, a_i)|$,
       which gives the first claim.
       If $(S,x)$ and $(T,y)$ are both on the cycle, we can map them 
       both onto each other by appropriate inputs,
       which implies $|S| = |T|$ by the aforementioned fact. $\qed$
      \end{proof}
      
\subsection{Proof of Lemma~\ref{lem:period_Ap_group} (See page~\pageref{lem:period_Ap_group})} 
\periodApgroup*
      \begin{proof} 
       Notation as in the statement of the Lemma.
       From Equation \eqref{eqn:def_seq_grid_aut_transition} of Definition \ref{def:sequ_grid_decomp_aut},
       we have $\delta(S, a_j^i) \subseteq \pi_1(\delta_p^{(j)}((S,x), a_j^i))$
       for all $i \ge 0$.
       So as $\delta(S, a_j^{L_S}) = \operatorname{id}_{|S}$, by Lemma \ref{lem:lcm_cycle_subset}, 
       this gives $S \subseteq T$. Hence, as $|S| = |T|$, we get $S = T$.
       Further as $m \ge I_i$, by Equation \eqref{eqn:def_unary_aut_transition} of 
       Definition \ref{def:sequ_grid_decomp_aut}, as $P$ divides $\lcm(P, L_S)$, we have $x = y$. 
       By Lemma \ref{lem:unary_period_divides} this implies
       that the period of $\mathcal A_p^{(j)}$ divides $\lcm(P, L_S)$. $\qed$
      \end{proof}
      
\subsection{Proof of Proposition~\ref{prop:bound_periods_Ap_group} (See page~\pageref{prop:bound_periods_Ap_group})} \boundperiodsApgroup*
        \begin{proof}
         Notation as in the statement of the Proposition.
         Set
         $$
          \mathcal P =\{ \mathcal A_q^{(j)} \mid p = q + \psi(b) \mbox{ for some } b \in \Sigma \}.
         $$
         Denote by $I$ the maximal index, and by $P$ the least common multiple of the periods,
         of the unary automata from $\mathcal P$.
         Let $(S, x) \in Q_p^{(j)}$ be any state from the cycle of $\mathcal A_p^{(j)}$.
         From Equation \eqref{eqn:def_unary_aut_transition} of Definition \ref{def:sequ_grid_decomp_aut},
         by inspecting the second ''counting'' component of the states, we see that the index 
         of $\mathcal A_p^{(j)}$ must be greater or equal than $I$. 
         Hence $x \ge I$. By Equation \eqref{eqn:def_unary_aut_transition}
         we have $(S, x) = \delta_p^{(j)}(s_p^{(0,j)}, a_j^x)$.
         Denote by $L_S$ the least common multiple of the cycle lengths of states from $S$
         with respect to the letter $a_j$, seen as a permutation of the states.
         Consider $(T, y) = \delta_p^{(j)}((S,x), a_j^{\lcm(P, L_S)})$.
         By Lemma \ref{lem:state_labels_grow_grp_case}, as they are both on the cycle, we
         have $|T| = |S|$. Then, using Lemma \ref{lem:period_Ap_group},
         we find $T = S$ and $x = y$, and that the period of $\mathcal A_p^{(j)}$
         divides $\lcm(P, L_S)$.
         Obviously, $L_S$ divides $L_j$.
         Inductively\footnote{See Remark \ref{rem:inductive_form} for an explantion
      of the induction scheme used.}, the periods of all automata in $\mathcal P$
         divide $L_j$, and so $P$ divides $L_j$.
         Hence, $\lcm(P, L_S)$ divides $L_j$. So the period of $\mathcal A_p^{(j)}$
         divides $L_j$. $\qed$
      \end{proof}
 
\subsection{Proof of Corollary~\ref{cor:Li_on_cycle} (See page~\pageref{cor:Li_on_cycle})} 
\Lioncycle*
        \begin{proof}
         We choose the same notation as in the statement of the Corollary.
         By Proposition \ref{prop:bound_periods_Ap_group}, the number $P$
         divides $L_j$.
         For any subset $R \subseteq Q$,
         if $L_R$ is the least common multiple of the cycle lengths of elements
         from $R$, it is also divided by $L_j$. Hence $\lcm(L_R, P)$
         divides $L_j$. 
         If $R, S \subseteq Q$ and $0 \le l \le L_j$,
         by Lemma \ref{lem:state_labels_grow_grp_case},
         if $(R, z) = \delta_p^{(j)}( (S,x), a_j^{l})$, then $|R| = |S|$,
         as $|S| \le |R| \le |T|$, and $|S| = |T|$ by assumption.
         Let $L_S$ be the least common multiple of the cycle lengths of elements in $S$.
         As $\lcm(L_S, P)$ divides $L_j$, we have, for
         \begin{equation}\label{eqn:lcm} 
          (R, z) = \delta_p^{(j)}( (S,x), a_j^{\lcm(L_S, P)})
         \end{equation}
         by the previous reasoning that $|R| = |S|$. 
         As, by assumption, $x \ge I$, 
         and as $(S, x) = \delta_p^{(j)}( s_p^{(0,j)}, a_j^x )$, we can use Lemma \ref{lem:period_Ap_group},
         which gives $R = S$ and $x = z$.
         But, as $\lcm(P, L_S)$ is a divisor of $L_j$, this
         gives $T = S$ and $x = y$ by repeatedly applying Equation \eqref{eqn:lcm} $\qed$
        \end{proof}
        
\subsection{Proof of Proposition~\ref{prop:index_specific_set} (See page~\pageref{prop:index_specific_set})}         
\indexspecificset*
       \begin{proof}
       
         Let $j \in \{1,\ldots, k\}$, $p \in H_j$ 
         and $\mathcal A_p^{(j)} = (\{a_j\}, Q_p^{(j)}, \delta_p^{(j)}, s_p^{(0,j)}, F_p^{(j)})$ the automaton 
         from Definition \ref{def:sequ_grid_decomp_aut}.
         Set
         $$
          \mathcal P =\{ \mathcal A_q^{(j)} \mid p = q + \psi(b) \mbox{ for some } b \in \Sigma \}.
         $$
         Denote by $I$ the maximal index and by $P$ the least common multiple of the periods
         of the unary automata in $\mathcal P$.
         If $\mathcal P = \emptyset$, by Definition~\ref{def:sequ_grid_decomp_aut} the 
         index of $\mathcal A_p^{(j)}$ is zero\footnote{More specifically, 
         in that case $p = (0,\ldots, 0)$ and the reachable part  from the start state
         of $\mathcal A_p^{(j)}$
         is essentially the reachable part of $\mathcal A$, by restricting to inputs from $\{a_j\}^*$.}.
         So the claim holds true as the state set labels are non-empty.
         So, suppose $\mathcal P \ne \emptyset$. We define a sequence $(T_n, y_n) \in Q_p^{(j)}$ of states for $n \in \mathbb N_0$.
         Set $(T_0, y_0) = \delta_p^{(j)}(s_p^{(0,j)}, a_j^I)$ (which implies $y_0 = I$) and
         $$
          (T_n, y_n) = \delta_p^{(j)}( (T_{j-1}, y_{j-1}), a_j^{L_j} ) 
         $$
         for $n > 0$.
         
         \begin{enumerate} 
         \item[(i)] \underline{Claim:} Let $(T, x)\in Q_p^{(j)}$ be some state from the cycle of $\mathcal A_p^{(j)}$.
           Then the state $\delta_p^{(j)}( s_p^{(0,j)}, a_j^{I + (|T| - |T_0|) L_j} )$
           is also from the cycle of $\mathcal A_p^{(j)}$.
           
           \medskip
           
            By construction, and Equation \eqref{eqn:def_unary_aut_transition} 
         from Definition \ref{def:sequ_grid_decomp_aut}, we have $y_n \ge I$ for all $n$.
         If $T_{n+1} \ne T_n$, then, by Corollary \ref{cor:Li_on_cycle}
         and Lemma \ref{lem:state_labels_grow_grp_case}, we have $|T_{n+1}| > |T_n|$.
         Hence, by finiteness, we must have a smallest $m$ such that $T_{m+1} = T_m$.
         By Corollary \ref{cor:Li_on_cycle}, this also implies $y_{m+1} = y_m$.
         Hence, we are on the cycle of $\mathcal A_p^{(j)}$, and the period
         of this automaton divides $L_j$ by Proposition \ref{prop:bound_periods_Ap_group}.
         This yields $(T_n, y_n) = (T_m, y_m)$ for all $n \ge m$.
         By Lemma \ref{lem:state_labels_grow_grp_case},
         the size of the state label sets on the cycle stays constant, and just
         grows before we enter the cycle.
         As we could
         add at most $|T_m| - |T_0|$ elements, and for $T_0, T_1, \ldots, T_{m}$
         each time at least one element is added, we have, as $m$
         was chosen minimal, that $m \le |T| - |T_0|$,
         where $T$ is any state label on the cycle, which all have the same cardinality $|T| = |T_k|$
         by Lemma \ref{lem:state_labels_grow_grp_case}.
         This means we could read at most $|T| - |T_0|$
         times the sequence $a_j^{L_j}$, starting from $(T_0, I)$,
         before we enter the cycle of $\mathcal A_p^{(i)}$.
         
         \item[(ii)] \underline{Claim:} We have $I \le (|T_0|-1) L_j$.
          
          Let $\mathcal A_q^{(j)} \in \mathcal P$ and suppose $p = q + \psi(b)$
          for $b \in \Sigma$. By Equation \eqref{eqn:def_seq_grid_aut_transition}
          from Definition \ref{def:sequ_grid_decomp_aut}, for 
          any $(S, x) = \delta_q^{(j)}( s_q^{(0,j)}, a_j^n )$
          we have $\delta(S, b) \subseteq \pi_1(\delta_p^{(j)}(s_p^{(0,j)}, a_j^n))$.
          In particular for $n = I$ we get $\delta(S, b) \subseteq T_0$, which gives $|S| \le |T_0|$,
          as $b$ induces a permutation on the states.
          Also for $n \ge I$ we are on the cycle of $\mathcal A_q^{(j)}$.
          Inductively\footnote{See Remark \ref{rem:inductive_form} for an explantion
      of the induction scheme used. The arguments in Claim (i) and (ii)
      would also work for $\mathcal P = \emptyset$ with a little adaption, but we prefered
      to state the case $\mathcal P = \emptyset$ here explicitly.}, then the index of $\mathcal A_q^{(j)}$ is at most $(|S|-1) L_j \le (|T_0|-1) L_j$.
          As $\mathcal A_q^{(j)} \in \mathcal P$ was chosen arbitrary,
          we get $I \le (|T_0|-1) L_j$.
         \end{enumerate}
         
         With Claim (ii) we can derive the upper bound $(|T|-1) L_j$ for the length of the word $a_j^{I + (|T| - |T_0|)L_j}$ 
         used in Claim (i), as 
         $$
          I + (|T| - |T_0|)L_j \le (|T_0|-1) L_j + (|T| - |T_0|) L_j = (|T|-1) L_j.
         $$
         And because (i) essentially says that the index of $\mathcal A_p^{(j)}$
         equals at most $I + (|T| + |T_0|) L_j$, this gives our bound for the index of $\mathcal A_p^{(j)}$. $\qed$
\end{proof}


\subsection{Some Additional Remarks and Alternative Proofs}

Here, I collect some remarks, or alternative proofs, that might be of additional interest. As far
as the paper is concerned, it is self-contained without the content of this section.

\begin{definition} The \emph{shuffle operation}, denoted by $\shuffle$, is defined as
 \begin{align*}
    u \shuffle v & := \left\{ \begin{array}{ll}
     \multirow{2}{*}{$x_1 y_1 x_2 y_2 \cdots x_n y_n  \mid$} &  u = x_1 x_2 \cdots x_n, v = y_1 y_2 \cdots y_n, \\ 
         &   x_i, y_i \in \Sigma^{\ast}, 1 \le i \le n, n \ge 1
  \end{array} \right\},
 \end{align*}
 for $u,v \in \Sigma^{\ast}$ and 
  $L_1 \shuffle L_2  := \bigcup_{x \in L_1, y \in L_2} (x \shuffle y)$ for $L_1, L_2 \subseteq \Sigma^{\ast}$.
\end{definition}

    %
    
    The languages accepted by the automata $\mathcal A_p^{(j)}$
    do not seem to play any role in the main part of the paper. But actually, some interesting relationships nevertheless
    hold true.
    
    \begin{lemma}
     Suppose $\Sigma = \{a_1, \ldots, a_k\}$. 
     Let $\mathcal A = (\Sigma, Q, \delta, s_0, F)$ be a finite automaton.
     Choose $j \in \{1,\ldots, k\}$ and $p \in H_j$, then
     $$
     L(\mathcal A_p^{(j)}) = \pi_{a_j}( \{ u \mid \exists w \in \Sigma\setminus\{a_j\}^* : u \in w \shuffle a_j^*, \psi(w) = p, u \in \perm(L(\mathcal A)) \})
     $$ 
     where $\pi_{a_j} : \Sigma^* \to \{a_j\}^*$ is given by $\pi_{a_j}(a_j) = a_j$ and $\pi_{a_j}(a_i) = \varepsilon$ for $i \ne j$.
    \end{lemma}
    \begin{proof} 
    By Proposition \ref{prop:grid_aut_decomp},  Proposition \ref{prop:conn_parikh} and Definition \ref{def:sequ_grid_decomp_aut}
    we have
    \begin{align*}
        a_j^n \in L(\mathcal A_p^{(j)})
         & \Leftrightarrow \pi_1(\delta_p^{(j)}(s_p^{(0,j)}, a_j^n)) \cap F  \ne \emptyset, p \in H_j \\
         & \Leftrightarrow \sigma_{\mathcal A}(p + \psi(a_j^n)) \cap F \ne \emptyset, p \in H_j \\
         & \Leftrightarrow p + \psi(a_j^n) \in \psi(L(\mathcal A)), p \in H_j \\ 
         & \Leftrightarrow a_j^n \in \pi_{a_j}(\{ u \mid \psi(u) - \psi(a_j^{|u|_{a_j}}) = p , u \in L(\mathcal A) \}) 
    \end{align*}
    And $\{ u \mid \psi(u) - \psi(a_j^{|u|_{a_j}}) = p , u \in L(\mathcal A) \}
    = \{ u \mid \exists w \in \Sigma^* : u \in w \shuffle a_j^*, \psi(w) = p \in H_j, u \in L(\mathcal A) \}
    = \{ u \mid \exists w \in \Sigma\setminus\{a_j\}^* : u \in w \shuffle a_j^*, \psi(w) = p, u \in \perm(L(\mathcal A)) \}$. $\qed$
    \end{proof}

    \begin{proposition} 
     Suppose $\Sigma = \{a_1, \ldots, a_k\}$. 
     Let $\mathcal A = (\Sigma, Q, \delta, s_0, F)$ be a finite automaton.
     Choose $j \in \{1,\ldots, k\}$,
     then\footnote{Note that $L \shuffle \emptyset = \emptyset$ for any language $L \subseteq \Sigma*$.}
     $$
      \perm(L(\mathcal A))
       = \bigcup_{p \in H_j} \bigcup_{\psi(w) = p} w \shuffle L(\mathcal A_p^{(j)}).
     $$
    \end{proposition}
    \begin{proof}
      Let $w \in \Sigma^*$ with $\psi(w) = (p_1, \ldots, p_k)$, then $w \in u \shuffle a_j^{p_j}$
      for some unique $u \in (\Sigma\setminus\{a_j\})^*$.
      Set $\overline p = (p_1, \ldots, p_{j-1}, 0, p_{j+1}, \ldots, p_k) \in H_j$.
      By Corollary \ref{cor:state_label_described_com_closure},
      Proposition \ref{prop:grid_aut_decomp}
      and Definition \ref{def:sequ_grid_decomp_aut} we have the equivalences:
      \begin{align*}
         w \in \perm(L(\mathcal A))
          & \Leftrightarrow \sigma_{\mathcal A}(\psi(w)) \cap F \ne \emptyset \\ 
          & \Leftrightarrow \pi_1(\delta_{\overline p}^{(j)}(s_{\overline p}^{(0,j)}, a_j^{p_j})) \cap F \ne \emptyset \mbox{ and } \psi(w) = \overline p + p_j \\
          & \Leftrightarrow a_j^{p_j} \in L(\mathcal A_{\overline p}^{(j)}) \mbox{ and } \psi(w) = \overline p + p_j \mbox{ and } \overline p \in H_j \\
          & \Leftrightarrow w \in u \shuffle L(\mathcal A_{\overline p}^{(j)}) \mbox{ and } \psi(u) = \overline p \in H_j.
          \quad \qed
      \end{align*}
    \end{proof}
    
     Also we note a connection with the class of jumping finite automata (JFAs).
Without giving a formal definition, let us mention here that a JFA looks syntactically indistinguishable from a classical (nondeterministic) finite automaton, only the processing of an input word is different: In one step, not necessarily the next symbol of the input is digested, i.e., the first symbol of the (remaining) input word, but just any symbol from the input. Pictorially speaking, the read head of the automaton might jump anywhere in the input before digesting the input symbol that it scans then; after digestion, this input symbol is cut out from the input.
      A JFA could be seen as a descriptional device
      for the commutative closure of some regular language. As
      noted in \cite{FernauHoffmann19}, a regular language is accepted by some jumping finite
      automaton if and only it is commutative.
      For a commutative regular language, we can operate a given deterministic automaton
      for the language as a jumping automaton, and hence the
      the size of a minimal jumping automaton is always smaller than
      the size of a minimal ordinary automaton.
      For the reverse direction, in the case of group languages, we can
      derive the next relative descriptional complexity result.
     
      \begin{restatable}[]{corollary}{relativedesccompjfa}\label{cor:relative_desc_compl_jfa}
       Let $L$ be a regular language accepted by some jumping finite automaton
       with $n$ states,
       which, if seen as an ordinary automaton, is a permutation automaton\footnote{Note that 
       by definition, a permutation automaton is deterministic.}. 
       Then $L$ could be accepted by some finite deterministic automaton
       of size at most $O((n e^{\sqrt{n \ln n}})^{|\Sigma|})$.
      \end{restatable}

\end{document}